\documentclass[11 pt]{article}
\usepackage{amsfonts,amsmath,color,amsthm,amssymb, url, enumerate, bbm, subfig}
\usepackage{amsfonts,amsmath,color,amssymb, url, enumerate, bbm, subfig}
\usepackage{graphicx}
\usepackage[DIV=14,BCOR=2mm,headinclude=true,footinclude=false]{typearea}
\usepackage[font=small,labelfont=bf]{caption}
\usepackage{hyperref}
\usepackage{tikz, etoolbox}
\usetikzlibrary{shapes}
\usetikzlibrary{arrows}
\usepackage[margin=1in]{geometry}

\newcommand{\al}{\alpha}

\newcommand{\f}{\frac}

\newcommand{\rar}{\rightarrow}

\newcommand{\N}{\mathbb{N}}
\newcommand{\R}{\mathbb{R}}

\newcommand{\bi}{\mathbbm{1}}
\newcommand{\LPOPT}{\text{OPT}_{\text{LP}}}
\newcommand{\VDV}{\text{VDV}}
\newcommand{\RPOPT}{\text{OPT}_{\text{Relaxed}}}

\definecolor{purple}{RGB}{128,0,128}

\def\qed{\vbox{\hrule\hbox{\vrule\kern3pt\vbox{\kern6pt}\kern3pt\vrule}\hrule}}

\newcommand{\TSPOPT}{\text{OPT}_{\text{TSP}}}

\newtheorem{theorem}{Theorem}[section]
\newtheorem{lemma}[theorem]{Lemma}
\newtheorem{corollary}[theorem]{Corollary}
\newtheorem{remark}[theorem]{Remark}
\newtheorem{proposition}[theorem]{Proposition}
\newtheorem{definition}[theorem]{Definition}
\newtheorem{claim}[theorem]{Claim}
\newtheorem{conjecture}[theorem]{Conjecture}

 \newtheorem*{prop*}{Proposition}

\theoremstyle{definition}

\title{Characterizing the Integrality Gap of the Subtour LP for the Circulant Traveling Salesman Problem}
\author{Samuel C. Gutekunst \\ \href{mailto:scg94@cornell.edu}{scg94@cornell.edu} \and David P. Williamson \\
\href{mailto:davidpwilliamson@cornell.edu}{davidpwilliamson@cornell.edu}
}
\date{}

\begin{document}
\maketitle

\begin{abstract}
We consider the integrality gap of the subtour LP relaxation of the Traveling Salesman Problem restricted to circulant instances.  De Klerk and Dobre \cite{Klerk11} conjectured that the value of the optimal solution to the subtour LP on these instances is equal to an entirely  combinatorial lower bound from  Van der Veen, Van Dal, and Sierksma \cite{VDV91}.  We prove this conjecture by giving an explicit optimal solution to the subtour LP.  We then show that the integrality gap of the subtour LP is  $2$ on circulant instances, making such instances one of the few non-trivial classes of TSP instances for which the integrality gap of the subtour LP is exactly known.   We also show that the degree constraints do not strengthen the subtour LP on circulant instances, mimicking the parsimonious property of metric, symmetric TSP instances shown in Goemans and Bertsimas \cite{Goe93} in a distinctly non-metric set of instances. 
\end{abstract}

\section{Introduction} 

The traveling salesman problem (TSP) is one of the most famous problems in combinatorial optimization.  An input to the TSP consists of a set of $n$ cities $[n]:=\{1, 2, ..., n\}$ and edge costs $c_{ij}$ for each pair of distinct $i, j \in [n]$ representing the cost of traveling from city $i$ to city $j$.  Given this information, the TSP is to find a minimum-cost tour visiting every city exactly once.    Throughout this paper, we implicitly assume that the edge costs are \emph{symmetric} (so that $c_{ij}=c_{ji}$ for all distinct $i, j\in [n]$)  and  interpret the $n$ cities as vertices of the complete undirected graph $K_n$ with edge costs $c_e=c_{ij}$ for edge $e=\{i, j\}$.  In this setting, the TSP is to find a minimum-cost Hamiltonian cycle on $K_n.$

With just this set-up, the TSP is well known to be NP-hard.  An algorithm that could approximate TSP solutions in polynomial time to within any constant factor $\alpha$ would imply P=NP (see, e.g., Theorem 2.9 in Williamson and Shmoys \cite{DDBook}).  Hence more restricted assumptions are placed on the edge costs.  If one assumes that edge costs are \emph{metric} (i.e. $c_{ij}\leq c_{ik}+c_{kj}$ for all distinct $i, j, k\in [n]$), it is known to be NP-hard to approximate TSP solutions in polynomial time to within any constant factor  $\alpha<\f{123}{122}$ (see Karpinski,  Lampis,  and Schmied \cite{Karp15}).  Conversely, the polynomial-time Christofides-Serdyukov  algorithm \cite{Chr76,Serd78} outputs a Hamiltonian cycle that is at most  a factor of $\f{3}{2}$ away from the optimal solution to any metric, symmetric instance.  

For  metric and symmetric edge costs, the Christofides-Serdyukov  algorithm remains the state of the art.  Significant work has gone into looking at more restricted sets of edge costs.   
For example, the $(1, 2)$-TSP restricts $c_{ij}\in\{1, 2\}$ for every edge $\{i, j\}$ (see, e.g., Papadimitriou and Yannakakis \cite{Pap93}, Berman and Karpinski \cite{Ber08}, Karpinski and Schmied  \cite{Karp12}).  In graphic TSP, instead, the input corresponds to a connected, undirected graph $G$ on vertex set $[n]$, and for $i, j\in [n],$ the cost $c_{ij}$ is the length of the shortest $i$-$j$ path in $G;$  approximation algorithms with stronger performance guarantees than the  Christofides-Serdyukov  algorithm are known in this case  (see, e.g.,  Oveis Gharan, Saberi, and Singh \cite{Gha11}, M{\"o}mke and Svensson \cite{Mom16}, Mucha \cite{Muc14}, and Seb{\H{o}} and Vygen \cite{Seb14}).    Yet another special case of metric and symmetric edge costs is Euclidean TSP, where each city $i\in[n]$ corresponds to a point $x_i\in\R^2$, and the cost $c_{ij}$ is given by the Euclidean distance between $x_i$ and $x_j;$ a polynomial-time approximation scheme is known in this case (see, e.g., Arora \cite{Aro96} and Mitchell \cite{Mitch99}).

In this paper, we consider a different class of instances: circulant TSP.  This class can be described by {\bf circulant matrices,}  matrices of the form
\begin{equation}\label{eq:circMat}
\begin{pmatrix} m_0 & m_1 & m_2 & m_3 & \cdots & m_{n-1} \\ m_{n-1} & m_0 & m_1 & m_2 & \cdots & m_{n-2} \\ m_{n-2} & m_{n-1} & m_0 & m_1 & \ddots & m_{n-3} \\ \vdots & \vdots & \vdots & \vdots & \ddots & \vdots \\ m_1 & m_2 & m_3 & m_4 & \cdots & m_0 \end{pmatrix} = \left(m_{(t-s) \text{ mod } n}\right)_{s, t=1}^n.
\end{equation}
In circulant TSP, the matrix of edge costs $C=(c_{i, j})_{i, j=1}^n$ is circulant; the cost of edge $\{i, j\}$ only depends on $i-j$ mod $n$.  Our assumption that the edge costs are symmetric and that $K_n$ is a simple graph implies that we can write our cost matrix in terms of $ \lfloor \frac{n}{2}\rfloor$ parameters: \begin{equation}\label{eq:ScircMat} C=(c_{(j-i) \text{ mod } n})_{i, j=1}^n=\begin{pmatrix} 0 & c_1 & c_2 & c_3 & \cdots & c_{1} \\ c_{1} & 0 & c_1 & c_2 & \cdots & c_2 \\ c_{2} & c_{1} &0 & c_1 & \ddots & c_{3} \\ \vdots & \vdots & \vdots & \vdots & \ddots & \vdots \\ c_1 & c_2 & c_3 & c_4 & \cdots & 0\end{pmatrix},\end{equation} with $c_0=0$ and  $c_i=c_{n-i}$ for $i=1, ..., \lfloor \frac{n}{2}\rfloor.$   Importantly, in circulant TSP we do not implicitly assume that the edge costs are also metric.  A {\bf circulant graph} is a graph whose weighted adjacency matrix is circulant.

Circulant matrices have well-studied structure (see, e.g., Davis \cite{Davis12} and Gray \cite{Gray06}), and form an intriguing class of instances for combinatorial optimization problems. They seem to provide just enough structure to make a compelling, ambiguous set of instances; it is unclear whether or not a given combinatorial optimization problem should remain hard or become easy when restricted to circulant instances.  Some classic combinatorial optimization problems become easy when restricted to circulant instances: in the late 70's, Garfinkel \cite{Gar77} considered a restricted set of circulant TSP instances motivated by minimizing wallpaper waste and argued that, for these instances, the canonical greedy algorithm for TSP (the nearest neighbor heuristic) provides an optimal solution.  
In the late 80's, Burkard and Sandholzer \cite{Burk91} showed that the decidability question for whether or not a symmetric circulant graph is Hamiltonian can be solved in polynomial time and showed that  bottleneck TSP is polynomial-time solvable on symmetric circulant graphs.  Bach, Luby, and Goldwasser (cited in Gilmore, Lawler, and Shmoys \cite{Gil85}) showed that one could find minimum-cost Hamiltonian paths in (not-necessarily-symmetric) circulant graphs in polynomial time.  In contrast, Codenotti, Gerace, and Vigna \cite{Code98} show that Max Clique and Graph Coloring remain NP-hard when restricted to circulant graphs and do not admit constant-factor approximation algorithms unless P=NP.

Because of this ambiguity, the complexity of circulant TSP has often been cited as an open problem (see, e.g., Burkhard \cite{Burk97}, Burkhard, De\u{\i}neko, Van Dal, Van der Veen, and Woeginger \cite{Burk98}, and Lawler, Lenstra, Rinnooy Kan, and Shmoys  \cite{Law07}). 
It is not known if the circulant TSP is solvable in polynomial-time or is NP-hard, even when restricted to instances where only two of the edge costs $c_1, ...., c_{\lfloor \frac{n}{2}\rfloor}$ are finite: the \emph{two-stripe circulant TSP}.  See Greco and Gerace \cite{Grec07} and Gerace and Greco \cite{Ger08b}.  Yang, Burkard, \c{C}ela, and Woeginger \cite{Yang97} provide a polynomial-time algorithm for asymmetric TSP in circulant graphs with only two stripes having finite edge costs.  The symmetric two-stripe circulant TSP is not, however, a special case of the asymmetric two-stripe version.
In addition to questions of minimizing wallpaper waste, circulant TSP has applications in  reconfigurable network design (see Medova \cite{Med93}). 

Motivated by positive results on Hamiltonicity and minimum-cost Hamiltonian paths, Van der Veen, Van Dal, and Sierksma \cite{VDV91} developed two heuristic algorithms for circulant TSP.  In the case where all costs $c_1, ..., c_{\lfloor \frac{n}{2}\rfloor}$  are distinct, one heuristic provides tours within a factor of two of the optimal solution.   In addition,  Van der Veen, Van Dal, and Sierksma \cite{VDV91} give an explicit combinatorial formula as a lower bound for circulant TSP.   Gerace and Greco \cite{Ger08} give a 2-approximation algorithm for the general case of circulant TSP when costs may not be distinct.  Gerace and Irving \cite{Ger98} give a $\frac{4}{3}$-approximation algorithm for circulant TSP when edge costs are also metric.  See also Greco and Gerace \cite{Gre08}.

De Klerk and Dobre \cite{Klerk11} consider several lower bounds for the circulant TSP, including the subtour elimination linear program (also referred to as the Dantzig-Fulkerson-Johnson relaxation \cite{Dan54} and the Held-Karp bound \cite{Held70}, and which we will refer to as the {\bf subtour LP}). Let $V=[n]$ denote the set of vertices in $K_n,$ and let $E$ denote the set of edges in $K_n$.  For $S\subset V$, denote the set of edges with exactly one endpoint in $S$ by $\delta(S):=\{e=\{i, j\}: |\{i, j\}\cap S|=1\}$ and let $\delta(v):=\delta(\{v\}).$  The subtour LP is:
\begin{equation}\label{eq:SLP}\begin{array}{l l l}
\min & \sum_{e\in E} c_e x_e & \\
\text{subject to} & \sum_{e\in \delta(v)} x_e = 2, & v=1, \ldots, n \\
& \sum_{e\in \delta(S)} x_e \geq 2, & S\subset V: S\neq \emptyset, S\neq V \\
&0\leq x_e \leq 1, & e\in E. \end{array}
\end{equation}
The constraints  $\sum_{e\in \delta(v)} x_e = 2$ are known as the degree constraints, while the constraints $\sum_{e\in \delta(S)} x_e \geq 2$ are known as the subtour elimination constraints.  When edge costs are metric (but not necessarily circulant), Wolsey \cite{Wol80},  Cunningham \cite{Cun86}, and Shmoys and Williamson \cite{Shm90} show that solutions to this linear program are within a factor of $\f{3}{2}$ of the optimal, integer solution to the TSP.

De Klerk and Dobre \cite{Klerk11} show that, in the context of circulant TSP, the subtour LP is at least as strong  as the combinatorial lower bound of  Van der Veen, Van Dal, and Sierksma \cite{VDV91}.  They also conjecture that, on any instance of circulant TSP, the combinatorial lower bound of  Van der Veen, Van Dal, and Sierksma \cite{VDV91} exactly equals the optimal solution to the subtour LP.

Our paper has two main results.  First, we prove the conjecture of De Klerk and Dobre \cite{Klerk11}.  Second, we show that the integrality gap of the subtour LP is 2 for circulant TSP instances, making such instances one of the few non-trivial classes of TSP instances for which the  integrality gap of the subtour LP is exactly known.  

We begin, in Section \ref{CTSP}, by reviewing major results and notation relevant to circulant TSP.  In Section \ref{Main}, we then state and prove our main theorem, showing that  the combinatorial lower bound of  Van der Veen, Van Dal, and Sierksma \cite{VDV91} exactly equals the optimal solution to the subtour LP.  In proving this result, we provide an  explicit optimal solution to the subtour LP on circulant instances.   As a corollary, we show that the degree constraints do not strengthen the subtour LP on circulant instances, mimicking the parsimonious property of metric, symmetric TSP instances shown in Goemans and Bertsimas \cite{Goe93} in a distinctly non-metric set of instances. In Section \ref{sec:main2} we complete our characterization of the integrality gap of the subtour LP and show that it is exactly 2 on circulant instances.   The instances we use to show that the integrality gap is 2 are the same instances for which the crown inequalities (a certain class of facet-defining inequalities for the metric, symmetric TSP; see Naddef and Rinaldi \cite{Nad92}) were derived.  We show that, unfortunately, adding the crown inequalities to the subtour LP does not reduce the integrality gap when restricted to circulant TSP instances.  This leads us to discuss and conjecture constraints whose addition to the subtour LP would lower its integrality gap on circulant instances.

 Our results serve to motivate circulant TSP as a non-trivial class of TSP instances for which there is substantial number-theoretic and combinatorial structure.  We hope our results reinvigorate broad interest in the circulant TSP,  and thus we conclude by indicating several compelling open questions.

\section{Circulant TSP: Notation and Background}\label{CTSP}
Throughout this paper, we consider circulant TSP instances where $V=[n]$ and let $d:=\lfloor\f{n}{2}\rfloor.$  We use $\equiv_n$ to denote the mod-$n$ equivalence relationship and assume all computations on the vertex set are done mod $n$. In circulant TSP, all edges $\{i, j\}$ such that $i-j\equiv_n k$ or $i-j\equiv_n (n-k)$ have the same cost $c_k.$  We refer to such edges as being in the  {\bf $k$-th stripe,} and we describe $k$ as the {\bf length} of the stripe.  Classic algorithms and bounds for circulant TSP depend only on the ordering of the stripes with respect to their costs.

\begin{definition}
Let $S\subset \{1, ..., d\}.$ The {\bf circulant graph $C\langle S\rangle$} is the (simple, undirected, unweighted) graph including exactly the edges associated  with the stripes $S$.  I.e., the graph with adjacency matrix $$A=(a_{ij})_{i, j=1}^n, \hspace{5mm} a_{ij} = \begin{cases} 1, & (i-j) \bmod n \in S \text{ or } (j-i) \bmod n \in S \\ 0, & \text{ else.} \end{cases}$$
\end{definition}
\noindent For a set of stripes $S$, the graph $C\langle S\rangle$ includes exactly the edges associated with those stripes.  Note that the adjacency matrix of a circulant graph is a symmetric circulant matrix (see Equations (\ref{eq:circMat}) and (\ref{eq:ScircMat})).

 Given such an input to circulant TSP, we associate a permutation $\phi:[d]\rar[d]$ that sorts the stripes in order of nondecreasing cost as well as a sequence that encodes the connectivity of $C\langle \{\phi(1), ..., \phi(k)\}\rangle$ for $1\leq k\leq d.$

\begin{definition}[Van der Veen, Van Dal, and Sierksma \cite{VDV91}]
Consider an instance of circulant TSP with edge costs $c_1, ..., c_d.$  A {\bf stripe permutation} $\phi: [d]\rar [d]$ is a permutation such that $c_{\phi(1)}\leq c_{\phi(2)} \leq \cdots \leq c_{\phi(d)}.$  The {\bf $g$-sequence associated to $\phi$} is $g^{\phi}=(g_0^{\phi}, g_1^{\phi}, ..., g_d^{\phi}),$ recursively defined by $$g_i^{\phi}=\begin{cases} n, & i=0 \\ \gcd\left(\phi(i), g_{i-1}^{\phi}\right), & \text{ else.} \end{cases}$$
\end{definition}
\noindent Proposition \ref{prop:Ham} will allow us to interpret $g_i^{\phi}$ as the number of components of  $C\langle \{\phi(1), ..., \phi(i)\}\rangle,$ the graph of all edges from the cheapest $i$ stripes.  See, e.g., Figure \ref{fig:CKNotation}.

Note that, if edge costs are not distinct for a given instance of circulant TSP, there may be multiple associated stripe permutations.  In this case, we will take $\phi$ to be an arbitrary stripe permutation sorting the costs.   In  Van der Veen, Van Dal, and Sierksma \cite{VDV91}, the $g$-sequence is denoted as $\left(\mathcal{G}\mathcal{C}\mathcal{D}(\phi(0)), ..., \mathcal{G}\mathcal{C}\mathcal{D}(\phi(d))\right)$ with $\phi(0):=n.$  In Greco and Gerace \cite{Grec07}, $\phi$ is referred to as a presentation.

An early result from  Burkard and Sandholzer \cite{Burk91} characterizes when Hamiltonian cycles exist in circulant graphs: Hamiltonian cycles exist whenever the an (undirected) circulant graph is connected.
\begin{proposition}[Burkard and Sandholzer \cite{Burk91}] \label{prop:Ham} 
Let $\{a_1, ..., a_t\}\subset [d]$ and let $\mathcal{G}=\gcd(n, a_1, ..., a_t).$ The circulant graph $C\langle \{a_1, ..., a_t\} \rangle$ has $\mathcal{G}$ components. The $i$th component, for $0\leq i\leq \mathcal{G}-1$, consists of $n/\mathcal{G}$ nodes $$\{i+\lambda \mathcal{G} \bmod n: 0\leq \lambda \leq \f{n}{\mathcal{G}}-1\}.$$  $C\langle \{a_1, ..., a_t\} \rangle$  is Hamiltonian if and only if $\mathcal{G}=1.$
\end{proposition}
Set  $$\ell :=  \min \{i: 1\leq i\leq d, g_i^{\phi}=1\}.$$ By Proposition \ref{prop:Ham}, the graph $C\langle\{\phi(1), ..., \phi(\ell-1)\}\rangle$ is not Hamiltonian, while $C\langle\{\phi(1), ..., \phi(\ell)\}\rangle$ is.  Hence any Hamiltonian tour uses an edge of cost at least $c_{\phi(\ell)},$ and tours can be constructed where $c_{\phi(\ell)}$ is the most expensive edge.  Thus this proposition not only resolves Hamiltonicity in circulant graphs, but it also resolves bottleneck TSP in circulant graphs.  In bottleneck TSP, the objective is to find a Hamiltonian tour for which the cost of the most expensive edge is minimized. Burkard and Sandholzer \cite{Burk91} use Proposition \ref{prop:Ham} to give a constructive algorithm for bottleneck TSP on circulant instances.  We will use  Proposition \ref{prop:Ham} to partition the vertices of circulant graphs.

 Moreover, Proposition \ref{prop:Ham} immediately gives rise to an easily solvable case of circulant TSP: if there exists a stripe permutation $\phi$ such that  $g^{\phi}_1=1,$ or equivalently, the length  $\phi(1)$ of a cheapest stripe is relatively prime to $n$.  For example, if $n$ is prime, circulant TSP is easily solvable: you obtain a Hamiltonian tour by following edges of the cheapest stripe; after $n$ edges you will have visited every node and returned to the start.  These observations were first made in Garfinkel \cite{Gar77}.

Proposition \ref{prop:Ham} can  be used to solve the minimum-cost Hamiltonian path problem on circulant instances.   
\begin{proposition}[Bach, Luby, and Goldwasser, cited in Gilmore, Lawler, and Shmoys \cite{Gil85}] \label{prop:HP}
Let $c_1, ..., c_d$ be the edge costs of a circulant instance and let $\phi$ be an associated stripe permutation.  The minimum-cost Hamiltonian path has cost $$\sum_{i=1}^{\ell} (g_{i-1}^{\phi}-g_i^{\phi})c_{\phi(i)}.$$
\end{proposition}
\noindent We sketch the proof in Appendix \ref{ap:HP}.

Proposition \ref{prop:HP} yields a natural lower bound on the optimal solution to circulant TSP instances: delete the most expensive edge of a Hamiltonian tour (of cost at least $c_{\phi(\ell)}$), and compare  the resultant Hamiltonian path to a minimum-cost Hamiltonian path.  
\begin{proposition}[Van der Veen, Van Dal, and Sierksma \cite{VDV91}]\label{prop:VDV}
Let $c_1, ..., c_d$ be the edge costs of a circulant instance and let $\phi$ be an associated stripe permutation.  
Any Hamiltonian tour costs at least
 $$\VDV:=\left(\sum_{i=1}^{\ell} (g_{i-1}^{\phi}-g_i^{\phi})c_{\phi(i)}\right) + c_{\phi(\ell)}.$$
 \end{proposition}
\noindent $\VDV$  is the aforementioned combinatorial lower bound for circulant TSP.

If there are multiple stripe permutations associated with an instance (i.e., the $c_i$ are not all distinct), the lower bound is independent of the stripe permutation chosen.  The lower bound is, moreover,  tight as can be shown by considering any instance where the cheapest stripe has length relatively prime to $n$.  For example the lower bound is tight for any instance where $\phi(1)=1.$

De Klerk and Dobre \cite{Klerk11} compare the $\VDV$ lower bound to several other well-known TSP bounds.  In a series of numerical experiments, they provide evidence to conjecture that the $\VDV$ lower bound is exactly equal to the value of the optimal solution to the  subtour LP (see Equation (\ref{eq:SLP})).

\begin{conjecture}[De Klerk and Dobre \cite{Klerk11}]  \label{conj:DK}
Let $c_1, ..., c_d$ be the edge costs of a circulant instance and let $\phi$ be an associated stripe permutation.  Let  $\LPOPT$  denote the optimal value of the subtour LP and $\VDV$ denote the value of the lower bound in Proposition \ref{prop:VDV}.  Then
$$\VDV= \LPOPT.$$
\end{conjecture}
\noindent Our first main result will be to prove this conjecture.

  De Klerk and Dobre \cite{Klerk11} provide further evidence for this conjecture by showing the following.
\begin{theorem}[De Klerk and Dobre \cite{Klerk11}]  \label{thm:DK}
Let $c_1, ..., c_d$ be the edge costs of a circulant instance and let $\phi$ be an associated stripe permutation.  Let  $\LPOPT$  denote the optimal value of the subtour LP and $\VDV$ denote the value of the lower bound in Proposition \ref{prop:VDV}.  Then:
$$\VDV\leq \LPOPT.$$
\end{theorem}

To prove this result, de Klerk and Dobre \cite{Klerk11} relax the subtour LP by dropping the degree constraints.  Denote by $\RPOPT$ the value of an optimal solution to this LP, so that:
$$\begin{array}{l l l l}
\RPOPT= & \min & \sum_{e\in E} c_e x_e & \\
& \text{subject to} & \sum_{e\in \delta(S)} x_e \geq 2, & S\subset V: S\neq \emptyset, S\neq V \\
& & 0\leq x_e \leq 1, & e\in E,
\end{array} $$
and $$\RPOPT \leq \LPOPT.$$  Any feasible solution to the dual of this relaxed LP thus also provides a lower bound on $\LPOPT.$  
De Klerk and Dobre \cite{Klerk11} provide a feasible solution to this dual of value equal to $\VDV,$ thus showing $$\VDV\leq \RPOPT\leq \LPOPT.$$

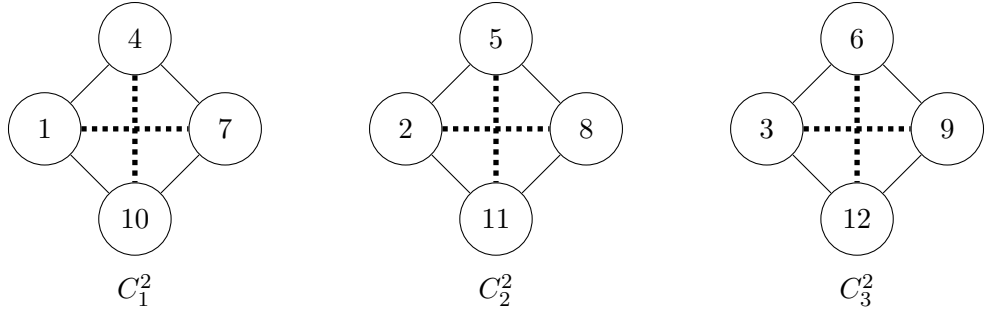
\begin{figure}[t]
\centering

\begin{tikzpicture}[scale=0.6]
\tikzset{vertex/.style = {shape=circle,draw,minimum size=2.5em}}
\tikzset{edge/.style = {->,> = latex'}}
\tikzstyle{decision} = [diamond, draw, text badly centered, inner sep=3pt]
\tikzstyle{sq} = [regular polygon,regular polygon sides=4, draw, text badly centered, inner sep=3pt]
\node[vertex] (a) at  (-10, 0) {$1$};
\node[vertex] (b) at  (-6, 0) {$7$};
\node[vertex] (c) at  (-8, 2) {$4$};
\node[vertex] (d) at  (-8, -2) {$10$};

\node[vertex] (a1) at  (-2, 0) {$2$};
\node[vertex] (b1) at  (2, 0) {$8$};
\node[vertex] (c1) at  (0, 2) {$5$};
\node[vertex] (d1) at  (0, -2) {$11$};

\node[vertex] (a2) at  (6, 0) {$3$};
\node[vertex] (b2) at  (10, 0) {$9$};
\node[vertex] (c2) at  (8, 2) {$6$};
\node[vertex] (d2) at  (8, -2) {$12$};
\draw[dotted,line width=2pt] (a) -- (b);
\draw[dotted,line width=2pt] (c) -- (d);
\draw (c) -- (a);
\draw (d) -- (a);
\draw (d) -- (b);
\draw (c) -- (b);

\draw[dotted,line width=2pt] (a1) -- (b1);
\draw[dotted,line width=2pt] (c1) -- (d1);
\draw (c1) -- (a1);
\draw (d1) -- (a1);
\draw (d1) -- (b1);
\draw (c1) -- (b1);

\draw[dotted,line width=2pt] (a2) -- (b2);
\draw[dotted,line width=2pt] (c2) -- (d2);
\draw (c2) -- (a2);
\draw (d2) -- (a2);
\draw (d2) -- (b2);
\draw (c2) -- (b2);

\draw  (-8, -3) node[below] {$C_1^2$};
\draw  (0, -3) node[below] {$C_2^2$};
\draw  (8, -3) node[below] {$C_3^2$};
\end{tikzpicture}

\caption{The graph $C\langle\{6, 3\}\rangle$ for $n=12.$  If $\{\phi(1), \phi(2)\}=\{3, 6\}$, the three components are $C^2_1, C^2_2,$ and $C^2_3.$  Dashed edges are of length 6.  In this example, $g_2=3.$ } \label{fig:CKNotation}
\end{figure}

Theorem \ref{thm:DK} leads to a bound on the  integrality gap of the subtour LP on circulant instances.  The integrality gap  represents the worst-case ratio of the original problem's optimal solution to the relaxation's optimal solution.
\begin{definition} Let $\TSPOPT(c_1, ..., c_d)$  denote the value of the optimal solution to the TSP for a circulant TSP instance with stripe costs $c_1, ..., c_d,$ and let 
$\LPOPT(c_1, ..., c_d)$ denote the value of the optimal solution of the subtour LP for the same circulant instance.   The {\bf integrality gap} for the subtour LP on circulant instances is  $$\sup_{(c_1, ..., c_d) \in \R^d_{\geq 0}} \f{\TSPOPT(c_1, ..., c_d)}{\LPOPT(c_1, ..., c_d)}.$$ 
\end{definition}  This ratio is bounded below by 1, since the subtour LP is a relaxation of the TSP.  
For metric (but not necessarily circulant) instances, Wolsey \cite{Wol80},  Cunningham \cite{Cun86}, and Shmoys and Williamson \cite{Shm90} show that the integrality gap of the subtour LP is at most $\f{3}{2}.$  
Theorem \ref{thm:DK} can also be used to show that, on circulant (but not necessarily metric) instances, the subtour LP also has a bounded integrality gap. 
\begin{theorem}\label{thm:int}
The integrality gap of the subtour LP restricted to circulant TSP instances is at most 2.  That is,
 $$\sup_{(c_1, ..., c_d) \in \R^d_{\geq 0}} \f{\TSPOPT(c_1, ..., c_d)}{\LPOPT(c_1, ..., c_d)}\leq 2,$$ 
\end{theorem}

\begin{proof}
Consider any circulant instance.  Let $\TSPOPT$ denote the value of the optimal solution to the TSP on this instance, $\LPOPT$ denote the value of the optimal solution to the subtour LP on this instance, and let $\VDV$ denote the value of the  Van der Veen, Van Dal, and Sierksma \cite{VDV91} lower bound on this instance.  By Theorem \ref{thm:DK},
 $$ \f{\TSPOPT}{\LPOPT}\leq \f{\TSPOPT}{\VDV}.$$
 Theorem 6.3 in Gerace and Greco \cite{Ger08} argues that $\f{\TSPOPT}{\VDV}\leq 2,$ by constructing Hamiltonian tours of cost at most $2\cdot\VDV.$ See Appendix \ref{ap} for details on this construction. \hfill 
\end{proof}

\section{A Combinatorial Interpretation of the Subtour LP}\label{Main}

In this section, we prove our first main result, answering Conjecture \ref{conj:DK}.  Recall that $$\ell =  \min \{i: 1\leq i\leq d, g_i^{\phi}=1\}.$$
\begin{theorem} \label{thm:main}
Let $c_1, ..., c_d$ be the edge costs of a circulant instance and let $\phi$ be an associated stripe permutation.  Let  $\LPOPT$  denote the optimal value of the subtour LP and let $\VDV$ denote the value of the lower bound in Proposition \ref{prop:VDV}.  Then:
$$\VDV = \LPOPT.$$  Moreover, an optimal solution to the subtour LP is achieved by setting, for $1\leq i\leq d$, the weight on  every edge $e$ of length $\phi(i)$ to be
  $$x_e = \begin{cases}
 \f{g^{\phi}_{i-1}-g_i^{\phi}}{n} ,& i\neq \ell, \phi(i)\neq \f{n}{2}\\
  2\f{g^{\phi}_{i-1}-g^{\phi}_i}{n} ,& i\neq \ell, \phi(i)= \f{n}{2} \\
 \f{g^{\phi}_{i-1}}{n} ,& i=\ell, \phi(i)\neq \f{n}{2} \\
 2\f{g^{\phi}_{i-1}}{n} ,& i= \ell, \phi(i)= \f{n}{2}. \\
  \end{cases}$$
\end{theorem}

The explicit $x_e$ values given in Theorem \ref{thm:main} spread out the weight placed by the Van der Veen, Van Dal, and Sierksma \cite{VDV91} bound,
 $$\VDV=\left(\sum_{i=1}^{\ell} (g_{i-1}^{\phi}-g_i^{\phi})c_{\phi(i)}\right) + c_{\phi(\ell)}.$$  The coefficient of 
$c_{\phi(i)}$ is spread uniformly over all edges of length $\phi(i).$  For $n$ even and $\phi(i)=d=n/2,$ there are only $\f{n}{2}$ such edges; otherwise there are $n$ edges.   
As a result, we remark the following.
  
  \begin{remark}
Let $x$ be defined as in Theorem \ref{thm:main}.  Then $$\sum_{e\in E} c_e x_e = \VDV.$$
\end{remark}

Note also that the solution places zero weight on edges of length $\phi(\ell+1), ..., \phi(d)$ as well as zero weight on edges of any length $\phi(i)$ such that $g_i^{\phi}=g_{i-1}^{\phi}.$  The optimal solution $x$, therefore, only depends on the relative ordering of edge costs $\phi$, and specifically, those stripes $\phi(i)$ for which $C\langle\{\phi(1), ..., \phi(i)\}\rangle$ has fewer components than $C\langle\{\phi(1), ..., \phi(i-1)\}\rangle$

  To simplify our work that follows, we assume that the edges are ordered so that \begin{equation}\label{eq:gdist}
  g_0^{\phi}>g_1^{\phi}>\cdots>g_{\ell}^{\phi}=1.
  \end{equation}    We can make this assumption without loss of generality: If $g_i^{\phi}=g_{i-1}^{\phi}$ for $i<\ell,$ then zero weight is placed on any edge of length $\phi(i)$ by both the Van der Veen, Van Dal, and Sierksma \cite{VDV91} bound and in the edge weights in Theorem \ref{thm:main}.  Both the   Van der Veen, Van Dal, and Sierksma \cite{VDV91} bound and the subtour LP solution we find in Theorem \ref{thm:main} thus remain the same on an instance where $c_{\phi(i)}$ is increased beyond $c_{\phi(\ell)}$.  By applying this argument iteratively, we can obtain an instance of circulant TSP for which the $g$-sequence is strictly decreasing until it reaches $1$, and  which the Van der Veen, Van Dal, and Sierksma \cite{VDV91} bound and the subtour LP  treat equivalently.

     For $0\leq i \leq \ell-1$ and $1\leq  k \leq g_i,$ we use $C_k^i$ to denote the vertex set of the $k$th connected component of the graph $C\langle\{\phi(1), ..., \phi(i)\}\rangle.$   Note that $C_k^i$ and $C_{k'}^i$ are isomorphic. See Figure \ref{fig:CKNotation}.  We let $C^i$ denote an arbitrary representative of $C_1^i, ..., C_{g_i}^i.$ 
     
  Our proof of Theorem \ref{thm:main} involves several steps.  In Lemma \ref{lem:deg}, we show that the solution $x$ posited satisfies the degree constraints.  We then characterize the components $C^i$ for $1\leq i\leq \ell-1$ as maximally dense:  in Lemma \ref{lem:Ci2} we show they satisfy the subtour elimination constraints with equality.  To complete the proof, we look at arbitrary subsets $S\subset V$ in Proposition \ref{prop:ST}. 
  
  Throughout the proof, we suppress the dependence of $g^{\phi}$ on $\phi$ to simplify notation.  It will be helpful to  treat our graph as a directed graph.    Each edge from the $i$th stripe, $i\neq n/2$, is directed $(v, v+i)$ (with the convention that $v+i$ is taken mod $n$).  If $n$ is even, we treat each edge of length $n/2$ incident to $v$ as two directed edges, $(v, v+(n/2))$ and $(v+(n/2), v)$, each of which is assigned half the weight of an edge with length $n/2.$ Thinking of our graph in this way means that every vertex $v$ is incident to exactly two edges from each stripe $i=1, ..., d$, with one edge directed into $v$ and one edge directed out of $v$.  That is, the edges of stripe $\phi(i)$ form a cycle cover on $V$.  Moreover, this simplifies the number of cases for $x_e$ since, if $n$ is even and $\phi(i)=n/2,$ we still spread the weight over $n$ edges; the weight on every edge $e$ of length $\phi(i)$ is then:
  $$x_e = \begin{cases}
 \f{g^{}_{i-1}-g_i^{}}{n} ,& i\neq \ell\\
 \f{g^{}_{i-1}}{n} ,& i=\ell. \\
  \end{cases}$$
   We fix $x\in\R^{E}$ to be the edge-weight vector with these weights.   
   
 For a set of edges $F\subset E,$ $x(F)$ denotes the total weight of edges in $F$: $\sum_{e\in F}x_e$.  We treat $\delta(S)$ as the set of all edges with exactly one endpoint in $S$, whether that edge is directed into or out of $S$.  Similarly, we treat $E(S)$ as the set of edges with both endpoints in $S,$ i.e. $E(S):=\{(i, j): i, j\in  S\}.$ For $A, B\subset V$, let $\delta^+(A, B):=\{e=(u, v): u\in A, v\in B\}$ denote the set of edges starting in $A$ and ending in $B$. We use $\sqcup$ to denote a disjoint union (i.e. a partition): $A=B\sqcup C$ means $A=B\cup C$ and $B\cap C=\emptyset.$   Finally, we use $\backslash$ for set-minus so that $A\backslash B=\{a\in A: a\notin B\}.$

  \begin{lemma}\label{lem:deg}
  For any vertex $v\in V$, $x(\delta(v))=2.$
  \end{lemma}

\begin{proof}
Let $i<\ell$ and consider edges of length $\phi(i)$ incident to $v$.   There are two edges of weight $\f{g_{i-1}-g_i}{n}:$ $(v, v+\phi(i))$ and $(v, v-\phi(i))$, so  the total weight of edges of length $\phi(i)$ incident to $v$  is $2\f{g_{i-1}-g_i}{n}.$  Analogously, the weight of edges of length $\phi(\ell)$ incident to $v$ is $\f{2g_{\ell-1}}{n}.$  Thus
\begin{align*}
x(\delta(v)) 
= \sum_{i=1}^{\ell} \sum_{\substack{e\in \delta(v):\\  \text{length}(e)=\phi(i)}} x_e
= \f{2}{n}\left(\left( \sum_{i=1}^{\ell-1}\left( g_{i-1}-g_i\right) \right)+  g_{\ell-1}\right)
= \f{2}{n}g_0 
=2,
\end{align*}
since $g_0=n.$  \hfill 
\end{proof}

We next argue that, for a set of vertices $S=C_k^i$, the only edges within $E(S)$ that have nonzero weight are those of length $\phi(1), ..., \phi(i).$  

\begin{lemma} \label{lem:Ci}
Let $S=C_k^i$ where $0\leq i \leq \ell-1$ and $1\leq  k \leq g_i.$  Let $e\in E(S)$.  Then $x_e>0$ implies $e$ is an edge in stripes $\phi(1), ..., \phi(i).$  
\end{lemma}
\begin{proof}
By Proposition \ref{prop:Ham},  $S=\{v:v\equiv_{g_i}j\}$ for some $0\leq j\leq g_i-1.$  Consider an edge of $e=(v, v+\phi(t))\in E(S)$ of length $\phi(t)$ with $t>i.$  Then, since $e$ has both endpoints in $C^i_k,$ $\phi(t)=c \cdot g_i$ for some $c\in \N.$  Hence  $g_t = \gcd(g_{t-1}, \phi(t))=\gcd(g_{t-1}, c\cdot g_i)=g_{t-1},$ since $g_{t-1}$ divides $g_i,$ and so $x_e=0.$
 \hfill 
\end{proof}

Lemma \ref{lem:Ci} lets us now show that the $C_k^i$ are maximally dense.

\begin{lemma}\label{lem:Ci2}
Let $S=C_k^i$ for $0\leq i \leq \ell-1$ and $1\leq  k \leq g_i.$  Then $x(\delta(S))=2.$
\end{lemma}

\begin{proof}
By Lemma \ref{lem:Ci}, we can compute $x(E(S))$ by only summing up the weights of edges in the cheapest $i$ stripes. 
Consider any fixed $j$ with $j\leq i<\ell.$  There are $n$ total edges of length  $\phi(j)$ and, since $j\leq i$, none of these edges are in any $\delta(C^i)$. Thus each isomorphic component $C^i\in \{C_1^i, ..., C_{g_i}^i\}$  has $\f{n}{g_i}$ edges of length $\phi(j)$ in $E(C^i)$, and each edge has  weight $\f{g_{j-1}-g_j}{n}.$  Hence $$\sum_{\substack{e\in E(S): \\ \text{length}(e)=\phi(j)}} x_e = \f{n}{g_i}\f{g_{j-1}-g_j}{n}=\f{g_{j-1}-g_j}{g_i}.$$ 

We can now compute:
\begin{align*}
x(E(S)) &= \sum_{j=1}^i \sum_{e\in E(S): \text{length}(e)=\phi(j)} x_e \\
&= \f{1}{g_i} \sum_{j=1}^i (g_{j-1}-g_j)\\
&= \f{g_0-g_i}{g_i} \\
&= \f{n}{g_i}-1\\
&= |C_k^i|-1.
\end{align*}

 The lemma then follows because the degree constraints imply that $x(\delta(S))+2x(E(S))=2|S|,$ so that $x(\delta(S))=2.$ \hfill 
\end{proof}

We now want to extend Lemma \ref{lem:Ci2} to show that $x(\delta(S))\geq 2$ for any $S\subset V$, not just those corresponding to components connected by a set of cheapest stripes.  We will consider any set $S^*$ and partition it into its intersections with certain $C^j$, where $S^*\subset \bigcup_{i=1}^s C_i^j.$  Expanding
$$x(E(S^*)) =  \sum_{i=1}^s x(E(S^* \cap C^j_i))+ \sum_{\substack{1\leq i_1, i_2 \leq s \\ i_1\neq i_2}}  x(\delta^+(S^*\cap C_{i_1}^j, S^* \cap C_{i_2}^j)),$$
we will bound each term of the sum.  To do so, we will bound $x(\delta^+(S^*\cap C_{i_1}^j, S^* \cap C_{i_2}^j))$ by $x(\delta^+( C_{i_1}^j,  C_{i_2}^j)).$  Our first step is thus to understand the edges between distinct $C^j$.

Proposition \ref{prop:Ham} implies that the vertices in a component $C^{j}$ are defined as $\{v: v \bmod g_{j} = i\}$ for some fixed $1\leq i\leq \f{n}{g_j}.$  Because $g_{j+1}$ divides $g_j$, $u\equiv_{g_j} v$ means that $u\equiv_{g_{j+1}} v$: if $u, v$ are in the same $C^j$ then $u, v$ are in the same $C^{j+1}.$  Consequently,  the edges of stripe $\phi(j+1)$ merge  $C^j$ into a smaller number of $C^{j+1}.$  The facts that the $C^i$ are all isomorphic and that $C^i$ has $g_i$ components implies that $\f{g_j}{g_{j+1}}$ components $C^j$ get merged into each $C^{j+1}$. See, for example, Figure \ref{fig:CiMerge}.
 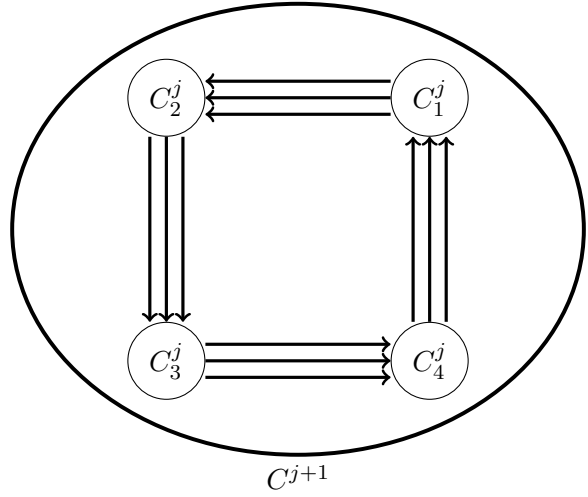
\begin{figure}[t]
  \tikzset{vertex/.style = {shape=circle,draw,minimum size=1.5em}}
\tikzset{edge/.style = {->,> = latex'}}
\tikzset{color1/.style = {fill=blue!65}}
\tikzset{set/.style={draw,circle,inner sep=0pt,align=center}}
\centering 
\begin{tikzpicture}[scale=1,transform shape]
    
        \draw[line width=1.5pt] (0.25, 1.75) ellipse (3.8cm and 3cm);
\draw  (0.25, -1.25) node[below] {$C^{j+1}$};

\node[vertex] (a) at  (2, 3.5) {$C_1^{j}$};
\node[vertex] (b) at  (-1.5, 3.5) {$C_2^{j}$};
\node[vertex] (c) at  (-1.5, 0) {$C_3^{j}$};
\node[vertex] (d) at  (2, 0) {$C_4^{j}$};

\draw [->, line width=1.2pt ]  (a) -- (b);
\draw [->, line width=1.2pt ]  ([yshift=0.22cm]a.west) -- node [above] {} ([yshift=0.22cm]b.east);
\draw [->, line width=1.2pt ]  ([yshift=-0.22cm]a.west) -- ([yshift=-0.22cm]b.east);

\draw [->, line width=1.2pt ]  (b) -- (c);
\draw [->, line width=1.2pt ]  ([xshift=0.22cm]b.south) -- node [above] {} ([xshift=0.22cm]c.north);
\draw [->, line width=1.2pt ]  ([xshift=-0.22cm]b.south) -- ([xshift=-0.22cm]c.north);

\draw [->, line width=1.2pt ]  (c) -- (d);
\draw [->, line width=1.2pt ]  ([yshift=0.22cm]c.east) -- node [above] {} ([yshift=0.22cm]d.west);
\draw [->, line width=1.2pt ]  ([yshift=-0.22cm]c.east) -- ([yshift=-0.22cm]d.west);

\draw [->, line width=1.2pt ]  (d) -- (a);
\draw [->, line width=1.2pt ]  ([xshift=0.22cm]d.north) -- node [above] {} ([xshift=0.22cm]a.south);
\draw [->, line width=1.2pt ]  ([xshift=-0.22cm]d.north) -- ([xshift=-0.22cm]a.south);
\end{tikzpicture}
\caption{The structure of edges from stripe $\phi(j+1)$ (marked by arrows) from Lemma \ref{lem:mergeCi}.  Here $\f{g_{j}}{g_{j+1}}=4.$ }\label{fig:CiMerge}
\end{figure}
Our next lemma describes the role of the edges from stripe $\phi(j+1)$ in this merging process.  It says that the subgraph of $C^{j+1}$ obtained by contracting each $C^j \subset C^{j+1}$ into a single vertex is a cycle.
\begin{lemma}\label{lem:mergeCi}
Suppose that $C^{j+1}=C^j_1 \sqcup \cdots \sqcup C^j_{\f{g_{j}}{g_{j+1}}}.$  Consider the directed graph $G'$ on $V'=\left[\f{g_{j}}{g_{j+1}}\right]$  where $(u, v) \in E'$ if and only if there is an edge of stripe $\phi(j+1)$ in $\delta^+(C_u^j, C_v^j).$ Then $G'$ is a directed cycle.
\end{lemma}

\begin{proof}
First, suppose that $u\in V$ is such that $u\in C^j_i$ (so that $i\in V'$).  For any other $v\in V$ with $v\in C^j_i,$ we have $u\equiv_{g_j} v$ and so  $u+\phi(j+1)\equiv_{g_j} v+\phi(j+1).$  Hence, the vertex $i\in V'$ has a single outgoing edge.  Analogously $u-\phi(j+1)\equiv_{g_j} v-\phi(j+1)$ so that the vertex $i\in V'$ has a single incoming edge.  These facts establish that every vertex of $G'$ has a single outgoing edge and a single incoming edge and $G'$ is a directed cycle cover.  However, $G'$ must also be connected: $C^{j+1}$ is a connected component of the graph $C\langle\{\phi(1), ..., \phi(j+1)\}\rangle.$  The only connected, directed cycle cover is a directed cycle.
\hfill
\end{proof}

Lemma \ref{lem:mergeCi} allows us to bound the total weight of edges of stripe $\phi(j+1)$ going between some $C^j$ in a $C^{j+1}.$ 

\begin{lemma}\label{lem:MultCi}
Suppose that $C^j_1, ..., C^j_s \subset C^{j+1}$ with $1<s\leq \f{g_{j}}{g_{j+1}}$ and $j<\ell.$ Provided $j<\ell-1$ or $ s< \f{g_{j}}{g_{j+1}},$
$$ \sum_{\substack{1\leq i_1, i_2 \leq s \\ i_1\neq i_2}} x(\delta^+(C^j_{i_1}, C^j_{i_2}))  \leq s -1.$$

\end{lemma}

\begin{proof}
By Lemma \ref{lem:Ci}, the only edges with both endpoints in $C^{j+1}$ with nonzero weight are those in stripes $\phi(1), ..., \phi(j+1).$  Moreover, any edge of stripe $\phi(i)$ with $i<j+1$ has both endpoints in the same $C^j$: $i\leq j$ implies $\phi(i)$ divides $g_j$, so $u+\phi(i)\equiv_{g_j} u;$  Proposition \ref{prop:Ham} implies that $u+\phi(i)$ and $u$ are in the same component of $C\langle\{\phi(1), ..., \phi(j)\}\rangle.$  Hence the only edges contributing to the sum $\sum_{1\leq i_1<i_2\leq s} x(\delta^+(C^j_{i_1}, C^j_{i_2})) $ are those from stripe $\phi(j+1).$ 

Consider the graph $G'$ from Lemma \ref{lem:mergeCi}, a cycle with vertices corresponding to $C^j_1, ..., C^j_{\f{g_{j}}{g_{j+1}}}.$  A subset of $s$ vertices of a cycle on ${\f{g_{j}}{g_{j+1}}}$ vertices contains at most $s-1+\bi_{\{s={\f{g_{j}}{g_{j+1}}}\}}$ edges, where $\bi_{\{\circ\}}$ denotes the indicator function that is $1$ if $\circ$ is true and $0$ otherwise.

Hence  at most $s-1+\bi_{\{s={\f{g_{j}}{g_{j+1}}}\}}$ terms in the sum $ \sum_{\substack{1\leq i_1, i_2 \leq s \\ i_1\neq i_2}} x(\delta^+(C^j_{i_1}, C^j_{i_2}))$ are nonzero.  Now consider any nonzero term $ x(\delta^+(C^j_{i_1}, C^j_{i_2}))\neq 0$.  Since the only edges contributing to this term are from stripe $\phi(j+1)$, we need only count the number of edges of stripe $\phi(j+1)$ starting in $C^j_{i_1}$ and ending in $C^j_{i_2}.$  There are $\f{n}{g_j}$ vertices in $C^j_{i_1}$, each of which has one outgoing edge of length $\phi(j+1)$ ending in  $C^j_{i_2}$. If $j<\ell -1$, each of these has weight $\f{g_{j}-g_{j+1}}{n}.$  Thus
\begin{align*}
 \sum_{\substack{1\leq i_1, i_2 \leq s \\ i_1\neq i_2}} x(\delta^+(C^j_{i_1}, C^j_{i_2})) &\leq \left(s-1+\bi_{\{s={\f{g_{j}}{g_{j+1}}}\}}\right)\f{n}{g_j}\f{g_{j}-g_{j+1}}{n}\\
&= \left(s-1+\bi_{\{s={\f{g_{j}}{g_{j+1}}}\}}\right)\left(1-\f{g_{j+1}}{g_j}\right).
\end{align*}
If $s\neq \f{g_{j}}{g_{j+1}}$, then the result follows because $\left(1-\f{g_{j+1}}{g_j}\right) \leq 1.$  Otherwise, when $s= \f{g_{j}}{g_{j+1}}$, the right hand side is $$\f{g_{j}}{g_{j+1}}\left(1-\f{g_{j+1}}{g_j}\right)=\f{g_{j}}{g_{j+1}}-1=s-1.$$
 
 The final case we must consider is when $j=\ell -1$ but $s< \f{g_{j}}{g_{j+1}}.$  In this case every edge of length $\phi(j+1)$ has weight $\f{g_{\ell-1}}{n}$ so that 
\begin{align*}
\sum_{1\leq i_1<i_2\leq s} x(\delta^+(C^j_{i_1}, C^j_{i_2})) &\leq \left(s-1+\bi_{\{s={\f{g_{j}}{g_{j+1}}}\}}\right)\f{n}{g_{\ell-1}}\f{g_{\ell-1}}{n}\\
&=(s-1)\f{n}{g_{\ell-1}}\f{g_{\ell-1}}{n}\\
&= s-1.
\end{align*}

\hfill
\end{proof}

\begin{proposition} \label{prop:ST}
Let $S\subset V$ ($2\leq |S|\leq n-2$).  Then $x(\delta(S))\geq 2.$
\end{proposition}

  \begin{figure}[t]
\centering 
\tikzset{vertex/.style = {shape=circle,draw,minimum size=1.5em}}
\tikzset{edge/.style = {->,> = latex'}}
\tikzset{color1/.style = {fill=blue!65}}
\tikzset{set/.style={draw,circle,inner sep=0pt,align=center}}
\begin{tikzpicture}[scale=0.9,transform shape]
    \draw[line width=1.7pt] (1, -1) rectangle (4.5,2);    
\draw  (4.1, 1.85) node[below] {$S^*$};

        \draw[line width=1.5pt] (2, 1) ellipse (0.5cm and 0.8cm);
\draw  (2, 0.2) node[below] {$C^j$};

        \draw[line width=1.5pt] (3.5, -1.2) ellipse (0.5cm and 0.8cm);
\draw  (3.5, -2) node[below] {$C^j$};

        \draw[line width=1.5pt] (6, -1.2) ellipse (0.5cm and 0.8cm);
\draw  (6, -2) node[below] {$C^j$};


        \draw[line width=1.5pt] (2.75, 0) ellipse (2.5cm and 3.5cm);
        \draw  (2.75, -3.5) node[below] {$C^{j+1}_k$};
\end{tikzpicture}
\caption{$S^*$ and choice of $j$ in Proposition \ref{prop:ST}.  In this example, $s=2$ of the $C^j$ have nonempty intersections with $S^*$.}\label{fig:propST1}
\end{figure}
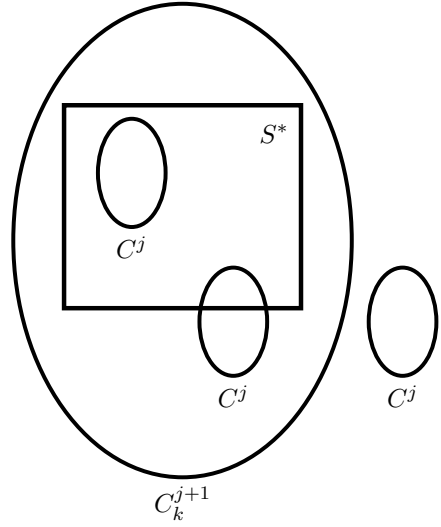

\begin{proof}

Using the fact that $x(\delta(S))+2x(E(S))=2|S|,$ it suffices to show that $x(E(S))\leq |S|-1$ for all $S$ (with $2\leq |S|\leq n-2$).  Suppose towards a contradiction that there is some $S^*$ with $x(E(S^*))>|S^*|-1$.  We consider three cases.

\noindent {\bf Case 1:} Suppose there exists such an $S^*$ that does not intersect with at least one $C^{\ell-1}.$  Then consider any such $S^*$ that is minimal by inclusion.   By Lemma \ref{lem:Ci2}, $S^* \neq C^i_k$ for any $0\leq i \leq \ell-1$ and $1\leq k\leq g_i$.  Since $S^*\subset C^{\ell}=V$ and the $C^i$ nest within the $C^{i+1}$, there are some $j$ and $k$ such that $S^* \subset C^{j+1}_k$ but $S^*$ is not contained in any single $C^j_1, ..., C^j_{g_j}$ (i.e. $j+1$ is the smallest value such that $S^*$ is properly contained in a $C^{j+1}$ which we denote $C^{j+1}_k)$.   See Figure \ref{fig:propST1}. 

Without loss of generality, suppose that the $C^j_i$ are labeled so that $C^j_1, ..., C_s^j $ have nonempty intersections with $S^*$ while $C_{s+1}^j, ..., C_{g_j}^j$ have empty intersections with $S^*.$    Note that, by choice of $j$, $C^j_1, ..., C_s^j \subset C^{j+1}_k$ so that $s\leq \f{g_j}{g_j+1}$.  We partition $S^*=\left(S^* \cap C_1^j\right) \bigsqcup \cdots \bigsqcup \left(S^* \cap C^j_s\right),$ so that 
\begin{align*}
x(E(S^*)) &=  \sum_{i=1}^s x(E(S^* \cap C^j_i))+  \sum_{\substack{1\leq i_1, i_2 \leq s \\ i_1\neq i_2}}  x(\delta^+(S^*\cap C_{i_1}^j, S^* \cap C_{i_2}^j)).
\intertext{By minimality of $S^*$}
x(E(S^*)) &\leq \sum_{i=1}^s \left(|S^*\cap C_i^j|-1\right)+  \sum_{\substack{1\leq i_1, i_2 \leq s \\ i_1\neq i_2}} x(\delta^+(S^*\cap C_{i_1}^j, S^* \cap C_{i_2}^j)) \\
&= |S^*|-s + \sum_{\substack{1\leq i_1, i_2 \leq s \\ i_1\neq i_2}}  x(\delta^+(S^*\cap C_{i_1}^j, S^* \cap C_{i_2}^j)).
\intertext{Expanding sets in the rightmost term}
x(E(S^*))&\leq |S^*|-s + \sum_{\substack{1\leq i_1, i_2 \leq s \\ i_1\neq i_2}}  x(\delta^+(C_{i_1}^j,C_{i_2}^j)).
\intertext{Note that our assumption that $S^*$ doesn't intersect with every single $C^{\ell-1}$ means that  Lemma \ref{lem:MultCi} applies, so that}
x(E(S^*)) &\leq |S^*|-s + (s-1) \\
&= |S^*|-1.
\end{align*}  
This contradicts our choice of $S^*$ as a counterexample.

The cases that remain are those where every single $S^*$  with $x(E(S^*))>|S^*|-1$ has $j=\ell-1$ (so that the smallest $C^{j+1}$ fully  containing $S^*$ is $C^{\ell}=V$) and $s=\f{g_{\ell-1}}{g_{\ell}}.$  This case means that $S^*$ has a nonempty intersection with every $C^{\ell-1}$. 

\noindent {\bf Case 2:} Suppose that there is some $C^{\ell-1}$ fully contained in $S^*$.  Then 
$2x(E(S^*))+x(\delta(S^*)) = 2|S^*|,$ so $x(E(S^*))>|S^*|-1$ implies $x(\delta(S^*))<2.$  Applying the same argument to $S^{*^c}:=V\backslash S^*$, we get $x(E(S^{*^c}))>|S^{*^c}|-1.$  But $S^{*^c}$ is entirely disjoint from at least one $C^{\ell-1}$, contradicting the assumption that case 1 does not apply.

  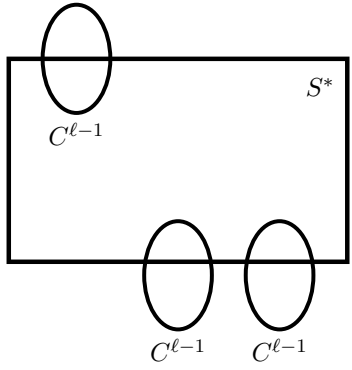
\begin{figure}[t]
\centering 
\tikzset{vertex/.style = {shape=circle,draw,minimum size=1.5em}}
\tikzset{edge/.style = {->,> = latex'}}
\tikzset{color1/.style = {fill=blue!65}}
\tikzset{set/.style={draw,circle,inner sep=0pt,align=center}}
\begin{tikzpicture}[scale=0.9,transform shape]
    \draw[line width=1.7pt] (1, -1) rectangle (6,2);    
\draw  (5.6, 1.85) node[below] {$S^*$};

        \draw[line width=1.5pt] (2, 2) ellipse (0.5cm and 0.8cm);
\draw  (2, 1.2) node[below] {$C^{\ell-1}$};

        \draw[line width=1.5pt] (3.5, -1.2) ellipse (0.5cm and 0.8cm);
\draw  (3.5, -2) node[below] {$C^{\ell-1}$};

        \draw[line width=1.5pt] (5, -1.2) ellipse (0.5cm and 0.8cm);
\draw  (5, -2) node[below] {$C^{\ell-1}$};
\end{tikzpicture}
\caption{Case 3 in the proof of Proposition \ref{prop:ST}.  $S^*$  intersects with every $C^{\ell-1}$ but does not fully contain any of the $C^{\ell-1}.$}\label{fig:c3}
\end{figure}

\noindent {\bf Case 3:} The only remaining case is that  $1\leq |S^*\cap C^{\ell-1}|< |C^{\ell-1}|$ for every $C^{\ell-1}.$  For this case we contradict that $x(E(S^*))>|S^*|-1$  by showing that $x(\delta(S^*))\geq 2.$  Since $g_{\ell-1}>g_{\ell}= 1,$ there are at least two $C^{\ell-1}$ and they are disjoint.  We will use the following claim to argue that each of them contributes at least $1$ to $x(\delta(S^*)).$

\begin{claim}\label{cm:ba}
Let $C$ be a set such that $x(\delta(C))=2$.  Suppose that $C=A\sqcup B$ where $x(\delta(A))\geq 2$ and $x(\delta(B))\geq 2.$  Then $x(\delta^+(A, B))+x(\delta^+(B, A))\geq 1.$
\end{claim}
This claim follows by expanding $\delta(A)$ and $\delta(B)$ and rearranging.
\begin{align*}
4 &\leq x(\delta(A)) + x(\delta(B))\\
&= \left( x(\delta^+(A, V\backslash C)) + x(\delta^+(V\backslash C, A)) + x(\delta^+(A, B))+x(\delta^+(B, A))\right) \\
& \hspace{10mm} + \left( x(\delta^+(B, V\backslash C)) + x(\delta^+(V\backslash C, B)) + x(\delta^+(A, B))+x(\delta^+(B, A))\right) \\
&=x (\delta(C))+2\left(x(\delta^+(A, B))+x(\delta^+(B, A))\right) \\
&= 2 + 2\left(x(\delta^+(A, B))+x(\delta^+(B, A))\right),
\end{align*}
from which the claim follows.

We now apply Claim \ref{cm:ba}.  Let  $C^{\ell-1}_i$ take the role of $C$, since by Lemma \ref{lem:Ci2} $x(\delta(C^{\ell-1}_i))=2.$  We partition $C^{\ell-1}_i=A\sqcup B$ where $A=S^*\cap C^{\ell-1}_i$ and $B=C^{\ell-1}_i\backslash A.$  Then $A, B\subset C^{\ell-1}_i$ and the fact that we are not in case 1 implies that $x(\delta(A))\geq 2$ and $x(\delta(B))\geq 2$ and the claim yields
$$x(\delta^+(S^*\cap C^{\ell-1}_i, C^{\ell-1}_i\backslash A))+x(\delta^+(C^{\ell-1}_i\backslash A, S^*\cap C^{\ell-1}_i)\geq 1.$$  All together,
\begin{align*}
x(\delta(S^*)) 
&\geq \sum_{i=1}^{g_{\ell-1}} x(\delta^+(S^*\cap C^{\ell-1}_i, C^{\ell-1}_i\backslash A))+x(\delta^+(C^{\ell-1}_i\backslash A, S^*\cap C^{\ell-1}_i) \\
&\geq \sum_{i=1}^{g_{\ell-1}}1 \\
&= g_{\ell-1} \geq 2. 
\end{align*}
Hence we contradict that $x(E(S^*)) >|S^*|-1$ and we have handled all cases.

\hfill
\end{proof}

\begin{proof}[Proof (Theorem \ref{thm:main})]
This proof follows immediately from Lemma \ref{lem:deg} and Proposition \ref{prop:ST}. \hfill 
\end{proof}

We note that Theorem \ref{thm:main}, together with the proof of Theorem \ref{thm:DK} in De Klerk and Dobre \cite{Klerk11}, indicate the following result.
\begin{corollary}
The degree constraints do not strengthen the subtour elimination LP for circulant TSP.   That is, letting $\RPOPT$ denote the value of an optimal solution to the subtour LP relaxation obtained by dropping the degree constraints, $$\RPOPT=\LPOPT=\VDV.$$
\end{corollary}

\begin{proof}
Our proof of Theorem \ref{thm:main} shows that $\RPOPT\leq \VDV,$ while the proof of Theorem \ref{thm:DK} shows that $\VDV\leq \RPOPT.$ \hfill 
\end{proof}

\section{The Integrality Gap of the Subtour LP}\label{sec:main2}

Theorem \ref{thm:main} allows us to exactly characterize the integrality gap of the subtour LP on circulant instances by considering the Van der Veen, Van Dal, and Sierksma \cite{VDV91} bound. 
In this section we provide an example showing that this bound can be off by a factor of 2 asymptotically.  This, together with Theorem \ref{thm:int}, will imply our second main theorem.
\begin{theorem}\label{main2}
The integrality gap of the subtour LP restricted to circulant instances is exactly 2.
\end{theorem}
The example we use to prove this theorem is intimately related to the crown inequalities for the  TSP, as we discuss in Section \ref{sec:conc}.

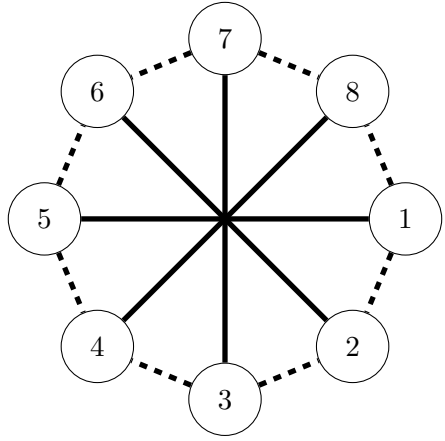
\begin{figure}[t]
\centering

\begin{tikzpicture}[scale=0.6]
\tikzset{vertex/.style = {shape=circle,draw,minimum size=2.5em}}
\tikzset{edge/.style = {->,> = latex'}}
\tikzstyle{decision} = [diamond, draw, text badly centered, inner sep=3pt]
\tikzstyle{sq} = [regular polygon,regular polygon sides=4, draw, text badly centered, inner sep=3pt]
\node[vertex] (a) at  (4, 0) {$1$};
\node[vertex] (b) at  (2.828, 2.828) {$8$};
\node[vertex] (c) at  (0, 4) {$7$};
\node[vertex] (d) at  (-2.828,2.828) {$6$};
\node[vertex] (e) at  (-4, 0) {$5$};
\node[vertex] (f) at  (-2.828, -2.828) {$4$};
\node[vertex] (g) at  (0, -4) {$3$};
\node[vertex] (h) at  (2.828, -2.828) {$2$};

\draw[dashed, line width=2pt] (a) to (b);
\draw[dashed, line width=2pt] (c) to (b);
\draw[dashed, line width=2pt] (c) to (d);
\draw[dashed, line width=2pt] (e) to (d);
\draw[dashed, line width=2pt] (e) to (f);
\draw[dashed, line width=2pt] (g) to (f);
\draw[dashed, line width=2pt] (g) to (h);
\draw[dashed, line width=2pt] (a) to (h);

\draw[line width=2pt] (e) to (a);
\draw[line width=2pt] (b) to (f);
\draw[line width=2pt] (c) to (g);
\draw[line width=2pt] (d) to (h);
\end{tikzpicture}
\caption{An example of a class of instances showing that the integrality gap of the subtour LP restricted to circulant instances is at least 2.  The dashed edges have weight $1/2$ and cost 1, while the full edges have weight 1 and cost 0.} \label{fig:intGap}
\end{figure}

\begin{proof}
Theorem \ref{thm:int} implies that the integrality gap is at most 2.  To prove the theorem it thus suffices to demonstrate an example where the Van der Veen, Van Dal, and Sierksma \cite{VDV91} bound is a factor of two away from the optimal TSP solution.  For such an example, we take $n=2^{k+1}$ so that $d=n/2=2^k$.  Suppose that $c_1=1, c_d=0,$ and $c_i>2^{k+1}$ otherwise.  Then $\phi(1)=d$ and $\phi(2)=1,$ so that $g^{\phi}_1=d,$ $g^{\phi}_i=1$ for $i\geq 2,$ and $\ell=1.$  By Theorem \ref{thm:main}, the optimal solution to the subtour LP has cost 
 $$\VDV=\left(\sum_{i=1}^{\ell} (g_{i-1}^{\phi}-g_i^{\phi})c_{\phi(i)}\right) + c_{\phi(\ell)}=d\cdot 0+d\cdot 1=d=2^k.$$  See Figure \ref{fig:intGap} for a picture of the corresponding subtour LP solution.

Now consider an optimal solution to the TSP.  It cannot use any edges other than those of lengths $1$ and $d$: we can find a tour of cost $2^{k+1}$ by just taking edges of length $1$ (i.e. $\{1, 2\}, \{2, 3\}, ..., \{n-1, n\}, \{n, 1\}$), while edges of any length other than $1$ or $d$ cost strictly greater than $2^{k+1}$.  Now consider any Hamiltonian cycle using only edges from these cheapest two stripes, and consider it as a directed cycle as in Proposition \ref{prop:ST}.  Suppose that it uses $s_1$ edges of length $1$ (where we interpret a directed edge $(u, u+1)$ as having length $1$), $s_{-1}$ edges of length $-1$ (where we interpret a directed edge $(u, u-1)$ as having length $-1$), and $n-s_1-s_{-1}$ edges of length $d=2^k.$  Because $n$ is even, there is no difference between an edge of length $d$ and $-d$: $v+d\equiv_n v-d.$

\begin{claim}\label{cm:modn}
Any Hamiltonian cycle satisfies $$s_1-s_{-1}\equiv_n 0.$$
\end{claim}

With the notation above, a Hamiltonian tour uses $n-s_1-s_{-1}$ edges of length $d.$  Since it starts and ends at the same vertex,\begin{equation}\label{eq:modn} 2^k(n-s_1-s_{-1})+s_1-s_{-1}\equiv_n0 \implies s_1-s_{-1}\equiv_n 2^k(s_1+s_{-1}).\end{equation}  Since $n$ and $2^k$ are even, $(2^k(s_1+s_{-1})) \bmod n$ is even; for the left and right sides to have the same parity,  $s_1-s_{-1}$ must therefore also be even. Moreover $$s_1+s_{-1}=s_1-s_{-1}+2s_{-1}$$ so that $s_1+s_{-1}$ is the sum of two even numbers and is therefore even.
Consider again Equation (\ref{eq:modn}).
Since $s_1+s_{-1}$ is even and $2^k=\f{n}{2},$ $$2^k(s_1+s_{-1})\equiv_n 0.$$ Thus Equation (\ref{eq:modn}) implies $$s_1-s_{-1}\equiv_n 0,$$ and Claim \ref{cm:modn} follows.

Since $s_1, s_{-1}\in [n],$ we have that $$s_1-s_{-1}\in \{-n, 0, n\}.$$  The cases where $|s_1-s_{-1}|=n$ imply a tour only using edges of length 1; i.e., a tour of cost $n=2^{k+1}.$  Thus we need only consider the case where $s_1=s_{-1}.$  Here we analogize an argument from Theorem 5.2 in Greco and Gerace \cite{Grec07}. 

 \begin{claim}\label{cm:compvisited}
A tour using just edges of lengths $1, -1$ and $d$ visits $\max\{s_1, s_{-1}\}+1$ components in of $C\langle\{2^k\}\rangle$.  Hence, a Hamiltonian tour requires $\max\{s_1, s_{-1}\}+1\geq \f{n}{2}=2^k.$
 \end{claim}

 We note that the graph $C\langle\{2^k\}\rangle$ using just edges of length $2^k$ has $2^k$ connected components $C_1^1, C_2^1, ..., C^1_{2^k}.$ We identify $C^1_i$ as consisting of the two vertices $\{i, 2^k+i\}$ connected by a single edge of length $2^k.$ 

Let $L=(e_1, ..., e_n)$ be a list of edges in any Hamiltonian tour using just edges of lengths $1, -1$ and $d$, so that $e_i\in \{-1, 1, d\}$ for $i=1, ..., n.$  From this list, we can bound the number of components of $C\langle\{2^k\}\rangle$ visited: first, we can delete any edges of length $d$: they do not cause us to change components   of $C\langle\{2^k\}\rangle;$  any length 1 edge connects $C_i^1$ to $C_{i+1}^1,$ while any length $-1$ edge connects $C_i^1$ to $C_{i-1}^1$ (regardless of whether or not any length $d$ edges are used).  Hence we need only consider the subsequence $L'$ of $L$ just consisting of edges of lengths 1 and $-1$ obtained by deleting the edges of length $d$.  Formally, $$L' = (e_{i_1}, ..., e_{i_k}): i_1<i_2<\cdots i_k, e_{i_j}\in \{\pm 1\}.$$

We upper bound the number of components of $C\langle\{2^k\}\rangle$ visited directly from $L'$ as follows: Set $U=1,$ corresponding to starting at some component.  Until $L'$ is either all $1$s or all $-1$s, find an occurrence of a  $1$ followed by a $-1$ in $L'$ (or a  $-1$ followed by a $1$); delete these two elements and increment $U$ by 1.  Once this process terminates, increment $U$ by $|L'|$ (the number of 1s or $-1$s remaining when $L'$ is either all 1s or all $-1$s).  Note that, at the end, $U=\max\{s_1, s_{-1}\}+1.$  $U$ provides an upper bound on the number of components of $C\langle\{2^k\}\rangle$ visited: Any time a 1 is followed by a $-1$ in $L$, the effect is to move from  $C_i^1$ to $C_{i+1}^1,$ then back to  $C_i^1.$  Hence we visit at most one new component,  $C_{i+1}^1$.  It is analogous any time a $-1$ is followed by a $1$.  Thus Claim \ref{cm:compvisited} holds.

 Since any Hamiltonian cycle must visit every component of $C\langle\{2^k\}\rangle$,  we need $$\max\{s_1, s_{-1}\}+1\geq \f{n}{2}=2^k.$$  That is, we need at least $2^k-1$ length $1$ edges, or $2^k-1$ length $-1$ edges, to connect all components.  
 
 Putting Claims \ref{cm:modn} and \ref{cm:compvisited} together, we find that we need $$s_1, s_{-1}\geq 2^k-1,$$  so that  $\TSPOPT\geq 2^{k+1}-2.$  We can find such a tour to establish equality: $$\{1, 2\}, \{2, 3\}, ..., \{2^k-1, 2^k\}, \{2^k, n\}, \{n, n-1\}, ..., \{2^{k}+2, 2^k+1\}, \{2^k+1, 1\}.$$  See, for example, Figure \ref{fig:TSPOPT}.
Thus 
$$\f{\TSPOPT}{\VDV} = \f{2^{k+1}-2}{2^k}\rar 2.$$

\begin{figure}[t]
\centering

\begin{tikzpicture}[scale=0.6]
\tikzset{vertex/.style = {shape=circle,draw,minimum size=2.5em}}
\tikzset{edge/.style = {->,> = latex'}}
\tikzstyle{decision} = [diamond, draw, text badly centered, inner sep=3pt]
\tikzstyle{sq} = [regular polygon,regular polygon sides=4, draw, text badly centered, inner sep=3pt]
\node[vertex] (a) at  (4, 0) {$1$};
\node[vertex] (b) at  (2.828, 2.828) {$8$};
\node[vertex] (c) at  (0, 4) {$7$};
\node[vertex] (d) at  (-2.828,2.828) {$6$};
\node[vertex] (e) at  (-4, 0) {$5$};
\node[vertex] (f) at  (-2.828, -2.828) {$4$};
\node[vertex] (g) at  (0, -4) {$3$};
\node[vertex] (h) at  (2.828, -2.828) {$2$};

\draw (a) to (h);
\draw (g) to (h);
\draw (g) to (f);
\draw[line width=2pt] (b) to (f);
\draw (b) to (c);
\draw (c) to (d);
\draw (e) to (d);
\draw[line width=2pt] (e) to (a);
\end{tikzpicture}
\caption{An optimal TSP solution for the instance in Theorem \ref{main2}.  Thick edges have cost 0 while thin edges have cost 1 so that this solution has cost $2^3-2=6.$} \label{fig:TSPOPT}
\end{figure}
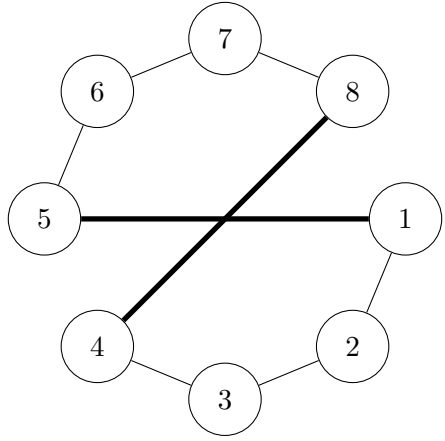 

\hfill 
\end{proof}

\section{Conclusions}\label{sec:conc}

Theorems \ref{thm:main} and \ref{main2} characterize the subtour LP when restricted to circulant instances: its optimal solution has an explicit combinatorial formulation given in Theorem \ref{thm:main} and is based entirely on how connectivity changes from $C\langle\{\phi(1), ..., \phi(i-1)\}\rangle$ to $C\langle\{\phi(1), ..., \phi(i)\}\rangle$ for stripes $i=1, ..., \ell.$  Moreover, the integrality gap  of the subtour LP on circulant instances is exactly two.

Our hope is also that this paper reinvigorates interest in several compelling open questions: What inequalities can be added to the subtour LP to strengthen its integrality gap on circulant instances?  Are there stronger linear programs for circulant instances, and do they translate to metric, symmetric TSP?  For example, de Klerk, Pasechnik, and Sotirov \cite{Klerk08} introduce a  semidefinite programming (SDP) relaxation for the TSP based on the theory of association schemes  (see also de Klerk, de Oliveira Filho, and Pasechnik \cite{Klerk12}). De Klerk and Dobre \cite{Klerk11} show that, for circulant instances, this SDP can be written as an LP; Gutekunst and Williamson \cite{Gut17} show that the general SDP has an unbounded integrality gap, but the integrality gap of the SDP (and equivalent LP) on circulant instances remains open.  Numerical experiments suggest that the integrality gap on circulant instances is at least 2, using the same  instances as in Section \ref{sec:main2}.  Similarly,  de Klerk and Sotirov \cite{Klerk12b} present an SDP that uses symmetry reduction to strengthen the SDP of  de Klerk, Pasechnik, and Sotirov \cite{Klerk08}; Gutekunst and Williamson \cite{Gut19} show that integrality gap of this SDP is also unbounded in general, but the gap is unknown when restricted to circulant instances.

With respect to adding inequalities to remove our bad instances (see Figure \ref{fig:intGap} for an example), we note that our instance achieving the worst-case integrality gap also appears in Naddef and Rinaldi \cite{Nad92} where they  construct explicit facet-defining inequalities that remove it from the subtour LP in a non-circulant setting.  These inequalities are the {\bf crown inequalities} and take the form $$\al^Tx\geq \al_0 := 12s(s-1)-2, \hspace{5mm} n=4s,$$ where the weight $\alpha_e$ that $\alpha$ places on edge $e$ is based only on the length of edge $e$: $$\al(v, v+j)=\begin{cases} 4s-6+j, & j<d \\ 2(s-1), &j=d. \end{cases}$$  Here, for example, the crown inequalities place a weight of $2(s-1)$ on each of the $d$ edges from the $d$th stripe, and a weight of $4s-5$ on each edge in the first stripe.  The subtour LP solution places a weight of $1$ on each of the $d$ edges of length $d$, and 1/2 on each of the $n$ length 1 edges.  Since $d=2s$:
$$\al^Tx =2s(2s-2)+\f{1}{2}4s(4s-5)=2s(6s-7)=12s^2-14s<\al_0=12s^2-12s-2$$ so that they are violated for any example where $n=4s$ and $s>1.$  

Unfortunately, adding these constraints does not reduce the integrality gap from 2.  We can instead consider solutions to the subtour LP that place marginally less weight on the $d$-edges and marginally more weight on the $1$-edges.  If we let $\lambda$ be the weight on the  $n$ edges of length 1 (on which $\alpha$ places weight $4s-5=n-5$), then $2-2\lambda$ is the weight on on each of the $d=\f{n}{2}$ edges of length $d$ (on which $\alpha$ places a weight of $2s-2=\f{n}{2}-2$).   The right hand side of the crown inequalities is $12\f{n}{4}\left(\f{n}{4}-1\right)-2=\f{3}{4}n^2-3n-2$, so we can solve for $$\lambda n (n-5)+(2-2\lambda)\f{n}{2} \left(\f{n}{2}-2\right)\geq \f{3}{4}n^2-3n-2 \rar \lambda \geq \f{n^2-4n-8}{2n^2-12n}=\f{1}{2}+ \f{2}{3n}+\f{1}{3(n-6)}$$ (assuming that $n>6$). Hence, setting $$ \lambda = \f{n^2-4n-8}{2n^2-12n}=\f{1}{2}+ \f{2}{3n}+\f{1}{3(n-6)}$$ suffices to find a solution that satisfies the subtour elimination constraints and the crown inequalities, but does not reduce the integrality gap.

\begin{proposition}
Adding the crown inequalities does not change the integrality gap of the subtour LP when restricted to circulant instances.
\end{proposition}

\begin{proof}
We take our solution above, setting $$ \lambda = \f{n^2-4n-8}{2n^2-12n}=\f{1}{2}+ \f{2}{3n}+\f{1}{3(n-6)}$$ and 
 placing a weight $\lambda$  on the $1$-edges (the dashed edges in Figure \ref{fig:intGap}) and $2-2\lambda$ on the edges of length $d$ (the full edges in Figure \ref{fig:intGap}).  Note that this solution is still feasible for the subtour LP: we are taking a convex combination of the instance in Theorem \ref{main2} and the Hamiltonian cycle using just 1-edges.  This thus lower bounds the integrality gap as:
 $$\f{\TSPOPT}{\LPOPT}=\f{n-2}{n\lambda}\rar 2$$ as $n\rar \infty,$ where $n=2^{k+1}.$
\hfill 
\end{proof} 

We note that the ladder and chain inequalities (see Boyd and Cunningham \cite{Boyd91}, Padberg and Hong \cite{Pad80}) 
can similarly be added to remove the solutions constructed in Theorem \ref{main2} but do not reduce the integrality gap from 2.

We conjecture that the following inequalities are valid.
\begin{conjecture}
The following inequality, if valid, would strengthen the subtour LP in the symmetric circulant case.  If $4|n,$ then
$$\sum_{i=1}^{n-1}\al_i \left(\sum_{\substack{e\in E:\\  \text{length}(e)=i}} x_e \right) \geq n-2, \hspace{5mm} \al_i = \begin{cases} i, & \text{ if i odd} \\d-i, & \text{ if i even.}\end{cases}$$
\end{conjecture}

Finally, as noted earlier, it is a major open question whether or not circulant TSP is polynomial-time solvable.  The answer is not known even in the case where only two stripes have finite cost.  It would be interesting to see if some of the tools developed recently for the metric TSP might be able to resolve this decades-long open question.

\section*{Acknowledgments}
 We thank Etienne de Klerk for pointing us to reference \cite{Klerk11}. We thank the referees for valuable comments and particularly thank the referee who suggested a cleaner proof of Proposition \ref{prop:ST}. 
 This work was supported by the Simons Institute for the Theory of Computing. 
 This material is also based upon work supported by the National Science Foundation Graduate 
Research Fellowship Program under Grant No. DGE-1650441. Any opinions,
findings, and conclusions or recommendations expressed in this material are those of the
authors and do not necessarily reflect the views of the National Science Foundation.

\clearpage

\clearpage

\bibliography{bibliog} 
\bibliographystyle{abbrv}

\clearpage

\appendix

\section{Appendix: Previous Results on Circulant TSP}

In this appendix, we first sketch the proof of Proposition \ref{prop:HP}.  We then sketch the 2-approximation algorithm for circulant TSP given in Gerace and Greco \cite{Ger08}.  

\subsection{Proof of Proposition \ref{prop:HP}}\label{ap:HP}
We recall Proposition \ref{prop:HP}.
\begin{prop*}[Proposition \ref{prop:HP}; from Bach, Luby, and Goldwasser, cited in Gilmore, Lawler, and Shmoys \cite{Gil85}]
Let $c_1, ..., c_d$ be the edge costs of a circulant instance and let $\phi$ be an associated stripe permutation.  The minimum-cost Hamiltonian path has cost $$\sum_{i=1}^{\ell} (g_{i-1}^{\phi}-g_i^{\phi})c_{\phi(i)}.$$
\end{prop*}

\begin{figure}[h!t]
\centering

\begin{tikzpicture}[scale=0.6]
\tikzset{vertex/.style = {shape=circle,draw,minimum size=2.5em}}
\tikzset{edge/.style = {->,> = latex'}}
\tikzstyle{decision} = [diamond, draw, text badly centered, inner sep=3pt]
\tikzstyle{sq} = [regular polygon,regular polygon sides=4, draw, text badly centered, inner sep=3pt]
\node[vertex] (a) at  (-11, 2) {$1$};
\node[vertex] (b) at  (-11, -2) {$7$};
\node[vertex] (c) at  (-7, 2) {$3$};
\node[vertex] (d) at  (-7, -2) {$9$};
\node[vertex] (a1) at  (-3, 2) {$5$};
\node[vertex] (b1) at  (-3, -2) {$11$};

\node[vertex] (c1) at  (3, 2) {$8$};
\node[vertex] (d1) at  (3, -2) {$2$};
\node[vertex] (a2) at  (7, 2) {$6$};
\node[vertex] (b2) at  (7, -2) {$12$};
\node[vertex] (c2) at  (11, 2) {$4$};
\node[vertex] (d2) at  (11, -2) {$10$};
\draw (a) -- (b);
\draw (c) -- (d);
\draw (a1) -- (b1);
\draw[line width=2pt] (b) -- (d);
\draw[line width=2pt] (c) -- (a1);

\draw (c1) -- (d1);
\draw (a2) -- (b2);
\draw (c2) -- (d2);
\draw[line width=2pt] (c1) -- (a2);
\draw[line width=2pt] (b2) -- (d2);

\draw[dotted,line width=2pt] (b1) -- (d1);
\end{tikzpicture}
\caption{Constructing a minimum-cost Hamiltonian Path via the nearest neighbor heuristic.  In this case, $n=12$, $\phi(1)=6$ (thin edges), $\phi(2)=2$ (thick edges), and $\phi(3)=3$ (dotted edges).  This process  fully connects a component of $C\langle\{\phi(1), ..., \phi(i)\}\rangle$, uses an edge of length $\phi(i+1)$ to move to the new component of  $C\langle\{\phi(1), ..., \phi(i)\}\rangle,$ and recursively fully connects that component.  When all possible edges of length $\phi(1), ..., \phi(i), \phi(i+1)$ have been added, the path connects a component of $C\langle\{\phi(1), ..., \phi(i+1)\}\rangle$ and the process repeats using edges of length $\phi(i+2)$.} \label{fig:HamPath}
\end{figure}
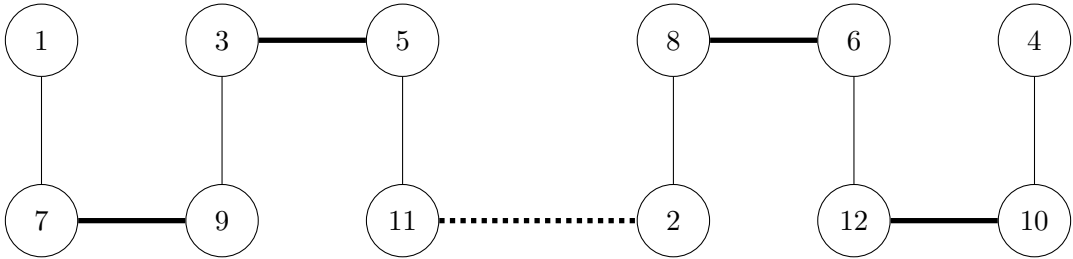

\begin{proof}[Sketch]
Van der Veen, Van Dal, and Sierksma \cite{VDV91} argue that the nearest neighbor heuristic\footnote{Start at some vertex and and follow a cheapest edge from that vertex.  Then, recursively grow a Hamiltonian path by adding a cheapest edge from the most recently added vertex to a vertex that has not yet been visited.}  constructs a Hamiltonian path using exactly $g_{i-1}^{\phi}-g_i^{\phi}$ edges from the $i$th cheapest stripe (see Figure \ref{fig:HamPath}).  This path thus has cost
 $$\sum_{i=1}^{\ell} (g_{i-1}^{\phi}-g_i^{\phi})c_{\phi(i)}.$$
 The optimality of such a path can be seen by applying Kruskal's algorithm \cite{Krusk56} for minimum-cost spanning trees:  For $1\leq i\leq \ell,$ Proposition \ref{prop:Ham} indicates that the graph $C\langle\{\phi(1), \phi(2), ..., \phi(i)\}\rangle$ has $g_i^{\phi}$ components.  Hence, at most $n-g_i^{\phi}$ edges can be used from the cheapest $i$ stripes without creating a cycle. Kruskal's algorithm will find a minimum-cost spanning tree using $n-g_1^{\phi}=g_0^{\phi}-g_1^{\phi}$ edges from the cheapest stripe, $g_1^{\phi}-g_2^{\phi}$ edges from the second cheapest stripe, and in general $g_{i-1}^{\phi}-g_i^{\phi}$ edges from the $i$th cheapest stripe.  This spanning tree thus also costs  $\sum_{i=1}^{\ell} (g_{i-1}^{\phi}-g_i^{\phi})c_{\phi(i)}.$  Since any Hamiltonian path  is itself a spanning tree,  any Hamiltonian path must cost at least this much; the constructed Hamiltonian path achieves this lower bound and is therefore optimal.   \hfill 
\end{proof}  

\subsection{Two Approximation for Circulant TSP}\label{ap}
This 2-approximation algorithm is motivated by a heuristic  Van der Veen, Van Dal, and Sierksma \cite{VDV91} developed for the case where every stripe has distinct cost.  The algorithm only adds edges of length $\phi(i)$ if $g_i^{\phi}<g_{i-1}^{\phi}.$  For simplicity of exposition, we'll suppress the dependence on $\phi$ and assume that $$n=g_0<g_1<g_2<\cdots<g_{\ell}=1$$ as in Section \ref{Main}.  

\subsubsection{Case 1: $g_{\ell-1}$ is Even}\label{ap1}

This algorithm is most straightforward when $g_{\ell-1}$ is even: First, it
 builds Hamiltonian paths on each component of  $C\langle\{\phi(1), ..., \phi(\ell-1)\}\rangle.$  It then deletes  one edge from  $g_{\ell-1}-1$ of these paths.  Finally, it adds $2(g_{\ell-1}-1)$ of length $\phi(\ell).$ 
 See Figure \ref{fig:PC2}. 

More specifically, construct a Hamilonian path on the vertices in the component of $C\langle\{\phi(1), ..., \phi(\ell-1)\}\rangle$ containing vertex $1$ using the nearest neighbor rule starting at vertex 1.  Call this path $P_1$ and let $z$ be the other endpoint of $P_1.$ Let $C_i^{\ell-1}$ be the component of $C\langle\{\phi(1), ..., \phi(\ell-1)\}\rangle$ containing vertex $1+(i-1)\phi(\ell)$ (as usual, here and throughout we implicitly consider all vertices mod $n$).  Translate $P_1$ to a Hamiltonian path $P_i$ on the vertices in $C_i^{\ell-1}$: add $(i-1)\phi(\ell)$ to the label of every vertex in $P_1$.   See Figure \ref{fig:PC}.

\begin{figure}[t]
\centering
\begin{tikzpicture}[scale=0.6]
\tikzset{vertex/.style = {shape=circle,draw,minimum size=2.5em}}
\tikzset{edge/.style = {->,> = latex'}}
\tikzstyle{decision} = [diamond, draw, text badly centered, inner sep=3pt]
\tikzstyle{sq} = [regular polygon,regular polygon sides=4, draw, text badly centered, inner sep=3pt]
\node[vertex] (a1) at  (-8, 6) {$1$};
\node[vertex] (a2) at  (-8, 3) {$13$};
\node[vertex] (a3) at  (-8, 0) {$19$};
\node[vertex] (a4) at  (-8, -3) {$7$};

\node[vertex] (b1) at  (-3, 6) {$6$};
\node[vertex] (b2) at  (-3, 3) {$18$};
\node[vertex] (b3) at  (-3, 0) {$24$};
\node[vertex] (b4) at  (-3, -3) {$12$};

\node[vertex] (c1) at  (2, 6) {$11$};
\node[vertex] (c2) at  (2, 3) {$23$};
\node[vertex] (c3) at  (2, 0) {$5$};
\node[vertex] (c4) at  (2, -3) {$17$};

\node[vertex] (d1) at  (7, 6) {$16$};
\node[vertex] (d2) at  (7, 3) {$4$};
\node[vertex] (d3) at  (7, 0) {$10$};
\node[vertex] (d4) at  (7, -3) {$22$};

\node[vertex] (e1) at  (12, 6) {$21$};
\node[vertex] (e2) at  (12, 3) {$9$};
\node[vertex] (e3) at  (12, 0) {$15$};
\node[vertex] (e4) at  (12, -3) {$3$};

\node[vertex] (f1) at  (17, 6) {$2$};
\node[vertex] (f2) at  (17, 3) {$14$};
\node[vertex] (f3) at  (17, 0) {$20$};
\node[vertex] (f4) at  (17, -3) {$8$};

\draw (a1) -- (a2);
\draw (a3) -- (a2);
\draw (a3) -- (a4);

\draw (b1) -- (b2);
\draw (b3) -- (b2);
\draw (b3) -- (b4);

\draw (c1) -- (c2);
\draw (c3) -- (c2);
\draw (c3) -- (c4);

\draw (d1) -- (d2);
\draw (d3) -- (d2);
\draw (d3) -- (d4);

\draw (e1) -- (e2);
\draw (e3) -- (e2);
\draw (e3) -- (e4);

\draw (f1) -- (f2);
\draw (f3) -- (f2);
\draw (f3) -- (f4);

\draw  (-8, -4) node[below] {$P_1$};
\draw  (-3, -4) node[below] {$P_2$};
\draw  (2, -4) node[below] {$P_3$};
\draw  (7, -4) node[below] {$P_4$};
\draw  (12, -4) node[below] {$P_5$};
\draw  (17, -4) node[below] {$P_6$};
\end{tikzpicture}
\caption{Translations of a Hamiltonian path $P_1$ to other components of $C\langle\{12, 6\}\rangle$ for a graph where $n=24, \phi(1)=12, \phi(2)=6,$ and $\phi(3)=5.$  In this example, $z=7.$} \label{fig:PC}
\end{figure}
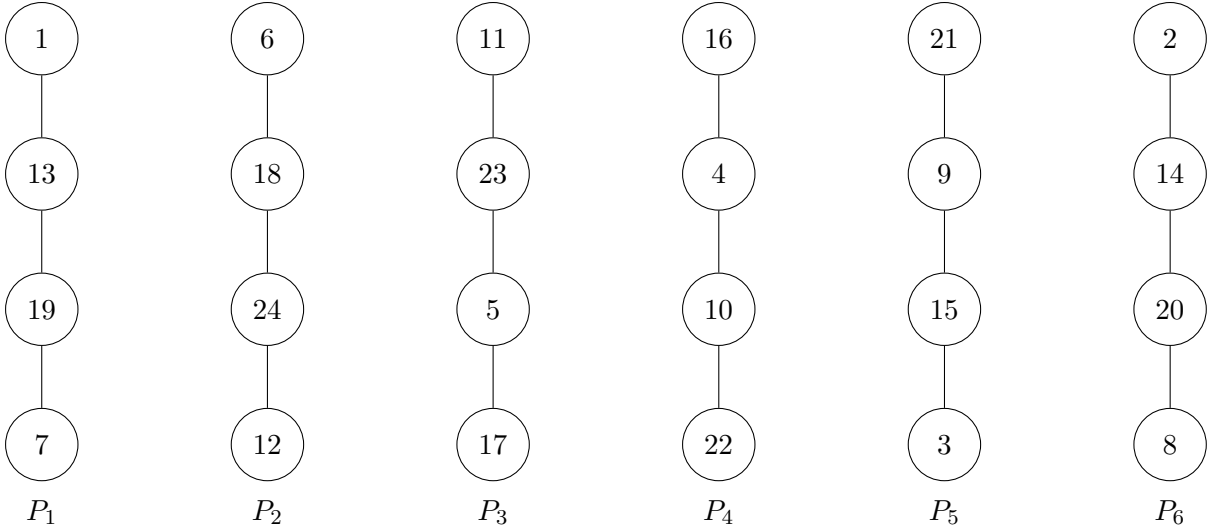

If $g_{\ell-1}$ is even, the algorithm deletes $g_{\ell-1}-2$ edges: pick some edge $\{u, v\}$ in $P_1.$  Delete the corresponding edge in each $P_2, P_3, ..., P_{g_{\ell-1}-1}$: delete the edge $\{u+(i-1)\phi(\ell), v+(i-1)\phi(\ell)\}$ from $P_i.$    Form a Hamiltonian cycle on the entire vertex set by adding $2(g_{\ell-1}-1)$ edges of length $\phi(\ell)$ as in Figure \ref{fig:PC2}\footnote{Specifically, add the following edges:
\begin{itemize}
\item Add the edges $\{1, 1+\phi(\ell)\}, \{1+2\phi(\ell), 1+3\phi(\ell)\}, ..., \{1+(g_{\ell-1}-2)\phi(\ell), 1+(g_{\ell-1}-1)\phi(\ell)\}.$  Also add the edges $\{z, z+\phi(\ell)\}, \{z+2\phi(\ell), z+3\phi(\ell)\}, ..., \{z+(g_{\ell-1}-2)\phi(\ell), z+(g_{\ell-1}-1)\phi(\ell)\}.$  This adds $g_{\ell-1}$ edges of length $\phi(\ell).$
\item Add the edges $\{u+\phi(\ell), u+2\phi(\ell)\}, \{u+3\phi(\ell), u+4\phi(\ell)\}, ..., \{u+(g_{\ell-1}-3)\phi(\ell), u+(g_{\ell-1}-2)\phi(\ell)\}.$  Also add the edges$\{v+\phi(\ell), v+2\phi(\ell)\}, \{v+3\phi(\ell), v+4\phi(\ell)\}, ..., \{v+(g_{\ell-1}-3)\phi(\ell), v+(g_{\ell-1}-2)\phi(\ell)\}.$ This adds $g_{\ell-1}-2$ edges of length $\phi(\ell).$
\end{itemize}}.

\begin{proposition}Consider any circulant instance where $g_{\ell-1}$ is even.  Let $\TSPOPT$ denote the optimal cost of a Hamiltonian tour on the circulant instance.  Then the above algorithm produces a Hamiltonian tour of cost at most $2\TSPOPT.$
\end{proposition}

\begin{proof}[Sketch]
By construction, the above algorithm produces a Hamiltonian tour.  We can analyze its cost in 3 steps:
\begin{enumerate}
\item When we start with $g_{\ell-1}$  paths (each Hamiltonian on a component of $C\langle\{\phi(1), ..., \phi(\ell-1)\}\rangle$), we have used all of the edges in a minimum-cost Hamiltonian path on $[n]$ except those of length $\phi(\ell).$  In total, these edges cost
 $$\sum_{i=1}^{\ell-1} (g_{i-1}^{\phi}-g_i^{\phi})c_{\phi(i)}.$$ 
\item We then delete some edges (translates of $\{u, v\}$), which cannot increase the cost.
\item Finally, we add $2(g_{\ell-1}-1)=2(g_{\ell-1}-g_{\ell})$ edges of cost $\phi(\ell).$ 
\end{enumerate}
Hence, we end with a tour costing at most 
 $$\sum_{i=1}^{\ell-1} (g_{i-1}^{\phi}-g_i^{\phi})c_{\phi(i)}+2(g_{\ell-1}-g_{\ell})c_{\phi(\ell)}\leq 2\sum_{i=1}^{\ell} (g_{i-1}^{\phi}-g_i^{\phi})c_{\phi(i)} \leq 2\TSPOPT.$$  The second  inequality follows because $\sum_{i=1}^{\ell} (g_{i-1}^{\phi}-g_i^{\phi})c_{\phi(i)}$ is the cost of a minimum-cost Hamiltonian path, which lower-bounds the cost of a Hamiltonian tour. \hfill
\end{proof}

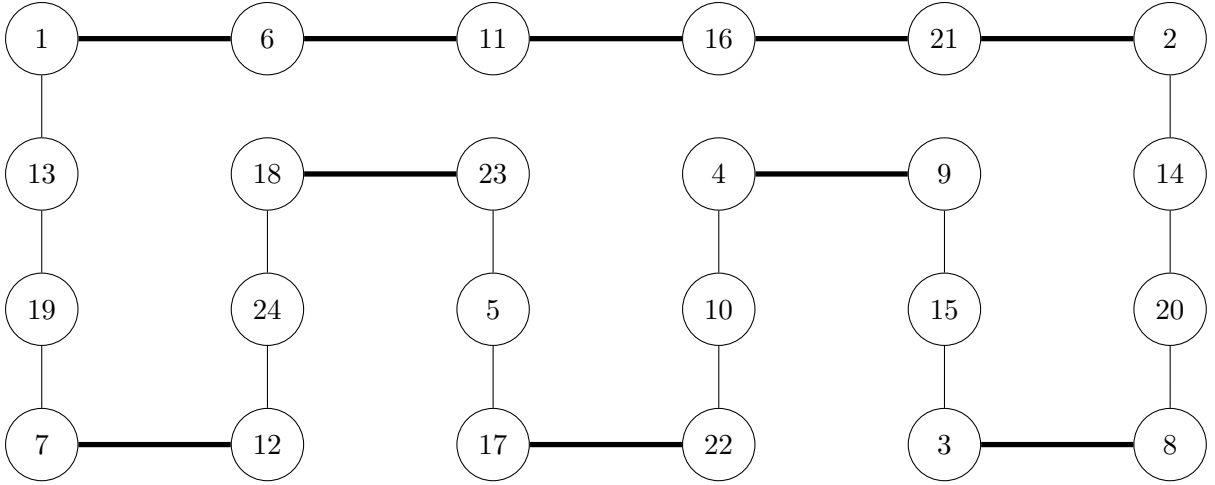
\begin{figure}[t]
\centering

\begin{tikzpicture}[scale=0.6]
\tikzset{vertex/.style = {shape=circle,draw,minimum size=2.5em}}
\tikzset{edge/.style = {->,> = latex'}}
\tikzstyle{decision} = [diamond, draw, text badly centered, inner sep=3pt]
\tikzstyle{sq} = [regular polygon,regular polygon sides=4, draw, text badly centered, inner sep=3pt]
\node[vertex] (a1) at  (-8, 6) {$1$};
\node[vertex] (a2) at  (-8, 3) {$13$};
\node[vertex] (a3) at  (-8, 0) {$19$};
\node[vertex] (a4) at  (-8, -3) {$7$};

\node[vertex] (b1) at  (-3, 6) {$6$};
\node[vertex] (b2) at  (-3, 3) {$18$};
\node[vertex] (b3) at  (-3, 0) {$24$};
\node[vertex] (b4) at  (-3, -3) {$12$};

\node[vertex] (c1) at  (2, 6) {$11$};
\node[vertex] (c2) at  (2, 3) {$23$};
\node[vertex] (c3) at  (2, 0) {$5$};
\node[vertex] (c4) at  (2, -3) {$17$};

\node[vertex] (d1) at  (7, 6) {$16$};
\node[vertex] (d2) at  (7, 3) {$4$};
\node[vertex] (d3) at  (7, 0) {$10$};
\node[vertex] (d4) at  (7, -3) {$22$};

\node[vertex] (e1) at  (12, 6) {$21$};
\node[vertex] (e2) at  (12, 3) {$9$};
\node[vertex] (e3) at  (12, 0) {$15$};
\node[vertex] (e4) at  (12, -3) {$3$};

\node[vertex] (f1) at  (17, 6) {$2$};
\node[vertex] (f2) at  (17, 3) {$14$};
\node[vertex] (f3) at  (17, 0) {$20$};
\node[vertex] (f4) at  (17, -3) {$8$};

\draw (a1) -- (a2);
\draw (a3) -- (a2);
\draw (a3) -- (a4);

\draw (b3) -- (b2);
\draw (b3) -- (b4);

\draw (c3) -- (c2);
\draw (c3) -- (c4);

\draw (d3) -- (d2);
\draw (d3) -- (d4);

\draw (e3) -- (e2);
\draw (e3) -- (e4);

\draw (f1) -- (f2);
\draw (f3) -- (f2);
\draw (f3) -- (f4);

\draw[line width=2pt] (a1) -- (b1);
\draw[line width=2pt] (a4) -- (b4);

\draw[line width=2pt] (c1) -- (d1);
\draw[line width=2pt] (c4) -- (d4);

\draw[line width=2pt] (e1) -- (f1);
\draw[line width=2pt] (e4) -- (f4);

\draw[line width=2pt] (b1) -- (c1);
\draw[line width=2pt] (b2) -- (c2);

\draw[line width=2pt] (d1) -- (e1);
\draw[line width=2pt] (d2) -- (e2);
\end{tikzpicture}
\caption{Constructing a Hamiltonian path when $g_{\ell-1}$ is even.  In this case, $n=24, \phi(1)=12, \phi(2)=6$ and $\phi(3)=5.$  We pick $\{u, v\}=\{1, 13\}.$} \label{fig:PC2}
\end{figure}

\subsubsection{Case 2: $g_{\ell-1}$ is Odd}\label{ap2}
If $g_{\ell-1}$ is odd, the algorithm of Gerace and Greco \cite{Ger08} proceeds similarly, but the analysis is more involved because the paths $P_1, ..., P_{g_{\ell-1}}$ cannot be connected into a Hamiltonian cycle as before.  Instead, the algorithm recursively calls itself to 
produce a Hamiltonian cycle $H$ in component $C_1^{\ell-1},$ as explained below.  As before, we take $P_1, ..., P_{g_{\ell-1}}$ to be Hamiltonian paths on the components of  $C\langle\{\phi(1), ..., \phi(\ell-1)\}\rangle,$ where the endpoints of $P_1$ are vertex 1 and vertex $z$, and each other $P_i$ is a translate of $P_1$.  We take edge $\{u, v\}$ of length $\phi(\ell-1)$ in path $P_1$.  Without loss of generality, we can assume $H$ contains edge $\{u, v\}$: $H$ contains some edge of length $\phi(\ell-1)$, and we can shift all the vertices in $H$ (adding some multiple of $g_{\ell-1}$ to each vertex) until that edge is $\{u, v\}.$

We then delete edge $\{u, v\}$ and its translates from $H, P_2, P_3, ..., P_{g_{\ell-1}-1}$ and add $2(g_{\ell-1}-1)$ edges of length $\phi(\ell)$ as in Figure \ref{fig:PC3}\footnote{Specifically:
\begin{itemize}
\item Add the edges $\{1+\phi(\ell), 1+2\phi(\ell)\}, \{1+3\phi(\ell), 1+4\phi(\ell)\}, ..., \{1+(g_{\ell-1}-2)\phi(\ell), 1+(g_{\ell-1}-1)\phi(\ell)\}.$  Also add the edges $\{z+\phi(\ell), z+2\phi(\ell)\}, \{z+3\phi(\ell), z+4\phi(\ell)\}, ..., \{z+(g_{\ell-1}-2)\phi(\ell), z+(g_{\ell-1}-1)\phi(\ell)\}.$ This adds $g_{\ell-1}-1$ edges of length $\phi(\ell).$
\item Add the edges $\{u+\phi(\ell), u+2\phi(\ell)\}, \{u+3\phi(\ell), u+4\phi(\ell)\}, ..., \{u+(g_{\ell-1}-3)\phi(\ell), u+(g_{\ell-1}-2)\phi(\ell)\}.$  Also add the edges$\{v+\phi(\ell), v+2\phi(\ell)\}, \{v+3\phi(\ell), v+4\phi(\ell)\}, ..., \{v+(g_{\ell-1}-3)\phi(\ell), v+(g_{\ell-1}-2)\phi(\ell)\}.$ This adds $g_{\ell-1}-1$ edges of length $\phi(\ell).$
\end{itemize}}.

\begin{figure}[t]
\centering
\begin{tikzpicture}[scale=0.6]
\tikzset{vertex/.style = {shape=circle,draw,minimum size=2.5em}}
\tikzset{edge/.style = {->,> = latex'}}
\tikzstyle{decision} = [diamond, draw, text badly centered, inner sep=3pt]
\tikzstyle{sq} = [regular polygon,regular polygon sides=4, draw, text badly centered, inner sep=3pt]
\node[vertex] (a1) at  (-8, 6) {$1$};
\node[vertex] (a2) at  (-8, 3) {$16$};
\node[vertex] (a3) at  (-8, 0) {$21$};
\node[vertex] (a4) at  (-8, -3) {$6$};
\node[vertex] (a5) at  (-8, -6) {$11$};
\node[vertex] (a6) at  (-8, -9) {$26$};

\node[vertex] (b1) at  (-3, 6) {$3$};
\node[vertex] (b2) at  (-3, 3) {$18$};
\node[vertex] (b3) at  (-3, 0) {$23$};
\node[vertex] (b4) at  (-3, -3) {$8$};
\node[vertex] (b5) at  (-3, -6) {$13$};
\node[vertex] (b6) at  (-3, -9) {$28$};

\node[vertex] (c1) at  (2, 6) {$5$};
\node[vertex] (c2) at  (2, 3) {$20$};
\node[vertex] (c3) at  (2, 0) {$25$};
\node[vertex] (c4) at  (2, -3) {$10$};
\node[vertex] (c5) at  (2, -6) {$15$};
\node[vertex] (c6) at  (2, -9) {$30$};

\node[vertex] (d1) at  (7, 6) {$7$};
\node[vertex] (d2) at  (7, 3) {$22$};
\node[vertex] (d3) at  (7, 0) {$27$};
\node[vertex] (d4) at  (7, -3) {$12$};
\node[vertex] (d5) at  (7, -6) {$17$};
\node[vertex] (d6) at  (7, -9) {$2$};

\node[vertex] (e1) at  (12, 6) {$9$};
\node[vertex] (e2) at  (12, 3) {$24$};
\node[vertex] (e3) at  (12, 0) {$29$};
\node[vertex] (e4) at  (12, -3) {$14$};
\node[vertex] (e5) at  (12, -6) {$19$};
\node[vertex] (e6) at  (12, -9) {$4$};

\draw[thick, densely dotted] (a1) -- (a2);
\draw[thick, densely dotted] (a3) -- (a4);
\draw[thick, densely dotted] (a5) -- (a4);
\draw[thick, densely dotted] (a5) -- (a6);
\draw[thick, densely dotted] (a1) to[bend right] (a6);

\draw (b1) -- (b2);
\draw (b3) -- (b4);
\draw (b5) -- (b4);
\draw (b5) -- (b6);

\draw (c1) -- (c2);
\draw (c3) -- (c4);
\draw (c5) -- (c4);
\draw (c5) -- (c6);

\draw (d1) -- (d2);
\draw (d3) -- (d4);
\draw (d5) -- (d4);
\draw (d5) -- (d6);

\draw (e1) -- (e2);
\draw (e3) -- (e2);
\draw (e3) -- (e4);
\draw (e5) -- (e4);
\draw (e5) -- (e6);

\draw[line width=2pt] (a3) -- (b3);
\draw[line width=2pt] (a2) -- (b2);

\draw[line width=2pt] (b1) -- (c1);
\draw[line width=2pt] (b6) -- (c6);

\draw[line width=2pt] (c3) -- (d3);
\draw[line width=2pt] (c2) -- (d2);

\draw[line width=2pt] (d1) -- (e1);
\draw[line width=2pt] (d6) -- (e6);

\end{tikzpicture}

\caption{The 2-approximation algorithm for circulant TSP when $g_{\ell-1}$ is odd.  In this case, $n=30, \phi(1)=15, \phi(2)=5$ and $\phi(3)=2.$  We find the Hamiltonian path $P_1= \{1, 16\}, \{16, 21\}, \{21, 6\}, \{6, 11\}, \{11, 26\}$ so that, e.g., $P_2=\{3, 18\}, \{18, 23\}, \{23, 8\}, \{8, 13\}, \{13, 28\}$ is the path translated by $1\times \phi(2)=2.$  We pick $\{u, v\}=\{18, 23\}$, and edge of length $\phi(2)=5.$  Since  $g_{\ell-1}=g_2=5$ is odd, we apply the recursive algorithm to find a Hamiltonian cycle on the vertices in $P_1$ (i.e., $C_1^2$). This yields the cycle $\{1, 6\}, \{6, 11\} \{11, 26\}, \{26, 21\}, \{21, 16\},\{16, 1\},$ including the edge $\{u, v\}$, so we don't need to shift it.  We then delete the $\{u, v\}$ and its translates from $H, P_2, P_3,$ and $P_4$ and reconnect using the thick edges (of length $\phi(3)=2$). Bolded edges are of length $\phi(\ell)$, while the dotted edges correspond to the edges from $H$ (after $\{u, v\}$ is removed).} \label{fig:PC3}
\end{figure}

This recursive process will eventually reach one of two halting conditions:
\begin{enumerate}
\item It  is called to find a Hamiltonian cycle on a component of $C\langle\{\phi(1), ..., \phi(t)\}\rangle$ where $\f{g_{t-1}}{g_t}$ is even, in which case it proceeds as in Case 1.  This cycle is then recursively used to create a Hamiltonian cycle on a component of $C\langle\{\phi(1), ..., \phi(t+1)\}\rangle$, and then on a component of $C\langle\{\phi(1), ..., \phi(t+2)\}\rangle$, and so on until it creates a Hamiltonian cycle on $C\langle\{\phi(1), ..., \phi(\ell-1)\}\rangle$ (following the process described above).  Note that  $\f{g_{t-1}}{g_t}$ counts the number of components of $C\langle\{\phi(1), ..., \phi(t-1)\}\rangle$ that get merged into a component of $C\langle\{\phi(1), ..., \phi(t)\}\rangle.$
\item Otherwise, we recursively call the algorithm until it attempts to produce a Hamiltonian cycle on a component of $C\langle\{\phi(1)\}\rangle$, in which case the  Hamiltonian cycle on $C\langle\{\phi(1)\}\rangle$ can be found by following edges of length $\phi(1)$ until a cycle is created.  In the case where $\phi(1)=n/2,$  we treat $\{1, 1+n/2\}$ as a cycle on $C\langle\{\phi(1)\}\rangle$ consisting of two length $d$ edges.
\end{enumerate}

\begin{proposition}Consider any circulant instance where $g_{\ell-1}$ is odd.  Let $\TSPOPT$ denote the optimal cost of a Hamiltonian tour on the circulant instance.  Then the above algorithm produces a Hamiltonian tour of cost at most $2\TSPOPT.$
\end{proposition}

\begin{proof}[Sketch]
By construction, the above algorithm produces a Hamiltonian tour.  We analyze its cost inductively at each stage of the recursion.  

Suppose the algorithm recurses until it finds a Hamiltonian cycle on a component of \\ $C\langle\{\phi(1), ..., \phi(t)\}\rangle$ (where possibly $t=1$).  We claim that the cost of the Hamiltonian cycle produced on this component is at most $$\f{2}{g_t}\sum_{i=1}^t (g_{i-1}-g_i)c_{\phi(i)}.$$   Indeed, if the algorithm halts because $t=1$, it produces a Hamiltonian cycle consisting of $\f{n}{g_1}$ edges of cost $c_{\phi(1)}$ and $$\f{n}{g_1}c_{\phi(1)}\leq 2\left(\f{n}{g_1}-1\right)c_{\phi(1)}=\f{2}{g_1}\sum_{i=1}^1 (g_{i-1}-g_i)c_{\phi(i)}.$$  If instead $t>1$,  we view the component of $C\langle\{\phi(1), ..., \phi(t)\}\rangle$ as the graph $C\langle\{\f{\phi(1)}{g_t}, \f{\phi(2)}{g_t}..., \f{\phi(t)}{g_t}\}\rangle$ with $\f{n}{g_t}$ vertices where edges of length $\f{\phi(i)}{g_t}$ have cost $c_{\phi(i)}$; since $g_t=\gcd(n, \phi(1), ..., \phi(t)),$ this is a well-defined circulant graph\footnote{Consider a component of   $C\langle\{\phi(1), ..., \phi(t)\}\rangle$ whose smallest vertex is labeled $i$.  Any vertex in this component with label $v$ can be relabeled with $\f{v-i}{g_t},$ which is an integer: $v, i$ in the same component of  $C\langle\{\phi(1), ..., \phi(t)\}\rangle$ implies $v\equiv_{g_t}i.$  Any edge in this component is of length $\phi(i)$ for $1\leq i\leq t,$ and $$u-v=\phi(t) \text{ if and only if } \f{u-i}{g_t}-\f{v-i}{g_t}=\f{\phi(t)}{g_t}.$$}.  Moreover, the algorithm reaching a base case of the recursion and $t>1$ implies that $\f{g_{t-1}}{g_t}$ is even, so that the graph $C\langle\{\f{\phi(1)}{g_t}, \f{\phi(2)}{g_t}..., \f{\phi(t-1)}{g_t}\}\rangle$ with $\f{n}{g_t}$ vertices has an even number of components.  Thus we can appeal to the analysis of the algorithm introduced in Appendix \ref{ap1} and, at the base case of recursion, the algorithm will produce a Hamiltonian tour on a component of $C\langle\{\phi(1), ..., \phi(t)\}\rangle$ of cost at most $$2\sum_{i=1}^t \f{g_{i-1}-g_i}{g_t}c_{\phi(i)}=\f{2}{g_t}\sum_{i=1}^t (g_{i-1}-g_i)c_{\phi(i)}.$$

We now analyze the algorithm inductively, claiming that at each subsequent iteration of the algorithm, it extends a Hamiltonian cycle on a component of 
 $C\langle\{\phi(1), ..., \phi(k)\}\rangle$ of cost at most $\f{2}{g_k}\sum_{i=1}^k (g_{i-1}-g_i)c_{\phi(i)}$ to a Hamiltonian cycle on a component of 
 $C\langle\{\phi(1), ..., \phi(k+1)\}\rangle$ of cost at most $\f{2}{g_{k+1}}\sum_{i=1}^{k} (g_{i-1}-g_i)c_{\phi(i)}.$  We do so in the following steps: 

\begin{enumerate}
\item By assumption, the Hamiltonian cycle on a component of $C\langle\{\phi(1), ..., \phi(k)\}\rangle$ costs at most $$\f{2}{g_k}\sum_{i=1}^k (g_{i-1}-g_i)c_{\phi(i)}.$$
\item There are $\f{g_k}{g_{k+1}}$ components of  $C\langle\{\phi(1), ..., \phi(k)\}\rangle$ that get joined into a component of  $C\langle\{\phi(1), ..., \phi(k+1)\}\rangle.$  The algorithm produces a minimum Hamiltonian path on the other $\f{g_k}{g_{k+1}}-1$ components of  $C\langle\{\phi(1), ..., \phi(k)\}\rangle$ that merge into $C\langle\{\phi(1), ..., \phi(k+1)\}\rangle.$  As in bounding the cost of base case of the recursion, each of these components is equivalent to the circulant graph $C\langle\{\f{\phi(1)}{g_k}, \f{\phi(2)}{g_k}..., \f{\phi(k)}{g_k}\}\rangle$ on $\f{n}{g_k}$ vertices so that the Hamiltonian path on each of these components will cost $$\sum_{i=1}^k \f{g_{i-1}-g_i}{g_k}c_{\phi(i)}.$$  These paths, with our Hamiltonian cycle, together cost at most
\begin{align*}
\f{2}{g_k}\sum_{i=1}^k (g_{i-1}-g_i)c_{\phi(i)} + \left(\f{g_{k}}{g_{k+1}}-1\right) \sum_{i=1}^k \f{g_{i-1}-g_i}{g_k}c_{\phi(i)} &=\left(\f{2}{g_k}+\f{1}{g_{k+1}}-\f{1}{g_k}\right)\sum_{i=1}^k (g_{i-1}-g_i)c_{\phi(i)} \\
&=\f{1}{g_{k+1}}\left(\f{g_{k+1}}{g_k}+1\right)\sum_{i=1}^k (g_{i-1}-g_i)c_{\phi(i)} \\
&\leq \f{2}{g_{k+1}}\sum_{i=1}^k (g_{i-1}-g_i)c_{\phi(i)},
\end{align*}
since $g_{k+1}\leq g_k.$
\item We then delete some edges, which cannot increase the cost.
\item Finally, we add $2\left(\f{g_k}{g_{k+1}}-1\right)$ edges of length $\phi(k+1)$ to form the Hamiltonian cycle on a component of $C\langle\{\phi(1), ..., \phi(k+1)\}\rangle.$  In total, these edges cost $$2\left(\f{g_k}{g_{k+1}}-1\right)c_{\phi(k+1)} = \f{2}{g_{k+1}}(g_k-g_{k+1})c_{\phi(k+1)}.$$
\end{enumerate}
Hence, we end with a Hamiltonian cycle on a component of  $C\langle\{\phi(1), ..., \phi(k+1)\}\rangle$  costing at most 
$$ \f{2}{g_{k+1}}\sum_{i=1}^k (g_{i-1}-g_i)c_{\phi(i)}+\f{2}{g_{k+1}}(g_k-g_{k+1})c_{\phi(k+1)} = \f{2}{g_{k+1}}\sum_{i=1}^{k+1} (g_{i-1}-g_i)c_{\phi(i)},$$
completing an inductive step.

Applying iteratively until we have a Hamiltonian cycle on the full instance, the total cost of this is at most 
$$ 2\sum_{i=1}^{\ell} (g_{i-1}^{\phi}-g_i^{\phi})c_{\phi(i)} \leq 2\TSPOPT.$$  The inequality again follows because $\sum_{i=1}^{\ell} (g_{i-1}^{\phi}-g_i^{\phi})c_{\phi(i)}$ is the cost of a minimum-cost Hamiltonian path, which lower-bounds the cost of a Hamiltonian tour. \hfill
\end{proof}

\end{document}